\documentclass[pra, showpacs, superscriptaddress, nofootinbib]{revtex4}
\usepackage{graphicx}
\usepackage{dcolumn}
\usepackage{bm}
\usepackage{multirow}
\usepackage{blkarray}
\usepackage{amssymb}
\usepackage{amsfonts}
\usepackage{amsmath}
\usepackage{amsthm}
\usepackage{color}
\usepackage{marvosym}
\usepackage{enumerate}
\newcommand{\mathbbm}[1]{\text{\usefont{U}{bbm}{m}{n}#1}}
\usepackage{tikz}
\usetikzlibrary{matrix,positioning,decorations.pathreplacing}

\newtheorem{theorem}{Theorem}

\newtheorem{definition}[theorem]{Definition}
\newtheorem{lemma}[theorem]{Lemma}

\newcommand{\bra}[1]{\langle #1|}
\newcommand{\ket}[1]{|#1\rangle}

\newcommand{\C}{\ensuremath{\mathbbm C}}
\newcommand{\be}{\begin{equation}}
\newcommand{\ee}{\end{equation}}
\newcommand{\bea}{\begin{eqnarray}}
\newcommand{\eea}{\end{eqnarray}}

\newcommand{\kommentar}[1]{}

\newcommand{\identity}{\mathbbm{1}}

\newcommand{\forget}[1]{}

\newcommand{\one}{\mbox{$1 \hspace{-1.0mm}  {\bf l}$}}

\begin{document}

\title{The Entanglement Hierarchy of $2\times m \times n$ Systems}

\author{M. Hebenstreit}
\affiliation{Institute for Theoretical Physics, University of Innsbruck, Technikerstr. 21A,  6020 Innsbruck, Austria}
\author{M. Gachechiladze}
\author{O. G\"uhne}
\affiliation{Naturwissenschaftlich-Technische Fakult\"at,
Universit\"at Siegen, Walter-Flex-Str. 3, 57068 Siegen, Germany}
\author{B. Kraus}
\affiliation{Institute for Theoretical Physics, University of Innsbruck, Technikerstr. 21A,  6020 Innsbruck, Austria}


\begin{abstract}
We consider three-partite pure states in the Hilbert space $\C^2 \otimes \C^m \otimes \C^n$ and investigate to which states a given state can be locally transformed with 
a non-vanishing probability. Whenever the initial and final states are elements of 
the same Hilbert space, the problem can be solved via the characterization of the 
entanglement classes which are determined via stochastic operations and classical communication (SLOCC). In general, there are infinitely many SLOCC classes. However, 
when considering transformations from higher- to lower-dimensional Hilbert spaces, 
an additional hierarchy among the classes can be found. This hierarchy of SLOCC classes coarse grains SLOCC classes which can be reached from a common resource state of higher dimension.
We first show that a generic 
set of states in $\C^2 \otimes \C^m \otimes \C^n$ for $n=m$ is the union of infinitely many SLOCC classes, which can be parameterized by $m-3$ parameters. However, for 
$n \neq m$ there exists a single SLOCC class which is generic. Using this result, 
we then show that there is a full-measure set of states in  $\C^2 \otimes \C^m \otimes \C^n$ such that any state within this set can be transformed locally to a full 
measure set of states in any lower-dimensional Hilbert space. 
We also investigate resource states, which can be transformed to any state (not excluding any zero-measure set) in the 
smaller-dimensional Hilbert space. We explicitly derive a state in 
$\C^2 \otimes \C^m \otimes \C^{2m-2}$ which is the optimal common resource of 
all states in $\C^2 \otimes \C^m \otimes \C^m$. We also show that for any $n < 2m$ it is impossible to reach all states in $\C^2 \otimes \C^m \otimes \C^{\tilde{n}}$ whenever $\tilde{n}>m$. 
\end{abstract}

\maketitle

\section{Introduction}

Entanglement, being a quantum mechanical feature without classical counterpart, 
has drawn a lot of attention in the last years. The existence of theoretical
protocols and practical applications in quantum information processing is mainly 
due to the subtle properties of multipartite entangled states. Thus, one of the 
main goals of quantum information theory is to gain a better understanding of 
the latter and an enormous amount of work has been devoted to the quantification 
and qualification of multipartite entanglement, for reviews see Refs.~\cite{hororeview, plenioreview, gtreview, siewertreview, zyczkowskireview}. Despite all the advances 
and the fact that bipartite entanglement is meanwhile well understood, we are still 
very far away form a complete characterization of multipartite entanglement.

In entanglement theory, quantum state transformations generated by local 
operations assisted by classical communication (LOCC) play a fundamental 
role. These are the transformations where each party is allowed to act locally
on its system, moreover, their actions can be correlated via classical 
communication. Entanglement is defined as the phenomenon that cannot be 
created with the restricted set of LOCC operations. In other words, entanglement 
theory is a resource theory where the free operations are LOCC transformations and the free states are the separable, i.e., the non-entangled states. Consequently, any entanglement measure has to be by definition non-increasing 
under LOCC.

If only two parties are considered, LOCC transformation can directly be 
characterized for pure states \cite{Nielsen}, but for mixed states many 
questions remain unclear \cite{kleinmann, ChLe12}. For the multipartite
case, even for pure states the characterization is difficult \cite{ChLe12}.
It is known that the set of LOCC transformations is strictly included in 
the mathematically simpler class of separable operations, but the latter is 
lacking a clear operational meaning. Moreover, it has been shown that the 
set of LOCC operations is not closed \cite{Ch11}. 

To gain nevertheless insight into multipartite entanglement, larger or smaller 
classes of operations have been considered. A smaller class of operations
are local unitary (LU) operations, and for these, equivalence of pure multi-qubit
states can directly be decided \cite{krausprl}. A larger class of operations
are stochastic LOCC (SLOCC) transformations. Here, one asks whether a state can be 
transformed into another locally, but requiring a non-vanishing success 
probability only \cite{DuVi00}. For pure states, one finds that two states are 
equivalent under SLOCC transformations, if there exists an invertible local 
operator which maps one state into the other. This leads to a coarse grained 
picture of state equivalences.
Moreover, a classification of entanglement obtained from single-party information only, so-called entanglement polytopes, has been studied \cite{WaDo13, SaOs14}.
 Whereas for three qubits there exist only two 
genuinely tripartite entangled SLOCC classes, the GHZ-class and the W-class \cite{DuVi00}, larger systems possess almost always infinitely many SLOCC classes \cite{VeDe02}. 

As there exist already infinitely many different SLOCC classes, an even more coarse grained classification of states is desireable.
The aim here is to group different SLOCC 
classes by considering resource states which belong to higher-dimensional 
Hilbert spaces. Let us explain the approach with a simple example. 
As mentioned above, there exist two fully entangled SLOCC classes for 
three qubits, the W-class and the GHZ-class. They are represented by the 
W-state 
\begin{equation}\label{W}
\ket{W} = \frac{1}{\sqrt{3}} (\ket{001}+\ket{010}+\ket{100})
\end{equation}
and the Greenberger-Horne-Zeilinger (GHZ) state
\begin{equation}\label{GHZ}
\ket{GHZ} = \frac{1}{\sqrt{2}} (\ket{000}+\ket{111}).
\end{equation}
Any fully entangled three-qubit state can be transformed via SLOCC into one
of these states, but these two states cannot be transformed into each other. 
One can ask, however, whether it is possible to find a common resource state 
for both classes by increasing the local dimension of one (or more) subsystems. 
A simple example would be to consider the state $\ket{\psi}=\ket{0}\ket{W}+\ket{1} \ket{GHZ}$ as a state in $\C^4\otimes \C^2 \otimes \C^2$. This state can be 
transformed into both classes, as Alice may make a $\sigma_z$-measurement on 
her ancilla qubit. 

The question we address here is whether such a state also exists in case infinitely 
many SLOCC classes exist and what is the smallest necessary increment in the dimension. Such a hierarchy of SLOCC classes has already been presented 
in Ref.~\cite{ChMi10}, where system sizes up to $\C^2 \otimes \C^3 \otimes \C^6$ have 
been investigated. Note that these considerations can lead to a kind of distance 
between SLOCC classes, as one may ask for the minimal dimension that is needed
to reach both SLOCC classes. States for which the local dimension of one Hilbert 
space needs to be increased only by one seem to be closer than those, for which 
a higher-dimensional Hilbert space is required. Note that there is an upper 
bound on the maximally required dimension, which is given by the increment of 
the dimension needed to perform teleportation. 

Here, we consider the three-particle case, where 
${\cal H}=\C^2\otimes \C^m\otimes \C^n$. Our choice of the system 
is motivated by the fact that it has been shown \cite{ChMi10} that 
for these dimensions, the SLOCC classification can be tackled with 
the help of so-called matrix pencils. A matrix pencil is essentially 
the set of all linear combinations of two matrices \cite{Ga59}.

We show that in order to obtain a generic set of states it is sufficient 
to increase the local dimension by just one. That is, a generic set of states 
in $\C^2\otimes \C^m\otimes \C^n$ is accessible by SLOCC operations from a state 
in $\C^2\otimes \C^m\otimes \C^{n+1}$. Surprisingly, there is not only one state 
which can be transformed into this generic set, but almost any state can be used 
for this purpose. More precisely, we show that there is a full-measure set of states in $\C^2\otimes \C^m\otimes \C^{n+1}$ such that any state in this set can be transformed 
into a generic set of states in $\C^2\otimes \C^m\otimes \C^n$. In order to derive 
this result, we characterize generic sets (in $\C^2\otimes \C^m\otimes \C^n$) of 
SLOCC classes of arbitrary dimensions. Let us mention here that we call a set 
generic if the complement of it is of lower dimension, i.e., if there exists a 
set of polynomial equations, in the coefficients of the state, which identifies 
the complement. Stated differently, a generic set is of full measure and almost 
all states belong to the generic set. We show that for $m=n$, the set of SLOCC 
classes, whose union leads to a generic set is parameterized by $m-3$ parameters. Interestingly, for the case that $m \neq n$ there exists a single SLOCC class 
which is generic. Moreover, any operator that can be applied by the party holding 
the qubit (Alice, $A$) can be performed locally by the other two. This resembles 
the bipartite case and is an interesting property when studying LOCC 
transformations.

The structure of the paper is the following. In Section II we first recall the 
relation between matrix pencils and three-partite pure states belonging to the 
Hilbert space $\C^{2}\otimes \C^{m}\otimes \C^{n}$. Then, we summarize some 
properties of matrix pencils. In particular we recall the so-called Kronecker 
normal form of matrix pencils, which can be considered as the analogue of 
the Jordan normal form for matrices. We then review the characterization of 
SLOCC classes in $2\times m\times n$ using matrix pencils and the conditions 
on more general state transformations \cite{ChMi10}. 

In Section III we consider generic states. We identify generic pencils and
the corresponding set of states, which are then of full measure. 
In Section IV we consider state transformations via local, however not 
invertible matrices. That is, we consider transformations from e.g. states in 
$2\times m\times n$ to states in $2\times m\times(n-1)$. We show there that a 
generic state in $2\times m\times n$ can be transformed into any state in a 
full measure set within $2\times m\times(n-1)$ for any $n$. As mentioned before, 
it turns out that this generic set of states (as the generic set of matrix pencils) 
is characterized by $m-3$ parameters in case $m=n$ and is a single SLOCC class in case 
$m \neq n$. 

In Section V we then study common resource states, i.e. states in $2\times m\times n$ which can be transformed into any state (not excluding zero measure sets) in a smaller-dimensional Hilbert space via 
local operations. We explicitly derive a state in $2\times m\times(2m-2)$ which is 
the common resource state of all states in $2\times m \times m$. We also show that
in $2\times m\times(2m-3)$ no common resource state for $2\times m \times m$ exists.
Finally, the Appendices contain the proofs of some statements in the main text 
and additional examples.

\section{Preliminaries}
In this section we introduce our notation and recall the relation between matrix pencils and three-partite states, where one of the constituent systems is a qubit. In particular, we review the findings presented in \cite{ChMi10} here. Throughout the paper we consider three--partite pure states belonging to the Hilbert space ${\cal H}=\C^{2}\otimes \C^{m}\otimes \C^{n}$, where w.l.o.g. $n\geq m$. Operators acting on the
first [held by Alice ($A$)], second [held by Bob ($B$)], and third system [held by Claire ($C$)] will be denoted by $A,B,C$ respectively.

\subsection{Linear Matrix pencils and their relation to pure states in $2\times m\times n$}

Two pure states $\ket{\psi}$ and $\ket{\phi}$ are said to be SLOCC--equivalent iff there exist local invertible operators $A$, $B$, and $C$ such that $\ket{\psi} = A\otimes B\otimes C \ket{\phi}$ \cite{DuVi00}. It can be easily seen that this indeed describes  an equivalence relation. In Ref. \cite{ChMi10} it has been shown that the SLOCC classes of pure states in $2\times m\times n$ can be characterized via linear matrix pencils. Moreover, a polynomial time algorithm deciding SLOCC equivalence of two pure states in $\mathcal{H}$ has been presented in Ref. \cite{ChMi10} \footnote{ In contrast to that, it has been shown that in case all systems belong to Hilbert spaces of dimension larger than two, then deciding whether two states are SLOCC equivalent or not is NP--hard \cite{Ha90}.}.

Any state, $\ket{\psi}$ in $\mathcal{H}$, can be written as
\begin{align}
	\ket{\psi} =  \ket{0}_A\ket{R}_{BC} + \ket{1}_A\ket{S}_{BC} =  \left[ \ket{0}_A (R \otimes \identity) + \ket{1}_A (S \otimes \identity) \right] \ket{\phi^+_n}_{BC} =  \left[ \ket{0}_A (\identity \otimes R^T) + \ket{1}_A (\identity \otimes S^T) \right] \ket{\phi^+_m}_{BC},
\end{align}
where $R$ and $S$ are complex $m \times n$ matrices and $\ket{\phi^+_k} = \sum_{i=0}^{k-1} \ket{i i}$ \footnote{Here and throughout the remainder of the paper, we ignore the normalization of states.}. To the pair of matrices $(R,S)$ a linear matrix pencil, $\mathcal{P}(R, S)$, which is a homogeneous matrix polynomial of degree 1 in variables $\mu$ and $\lambda$, i.e., $\mathcal{P}(R, S)=\mu R+\lambda S$ can be associated. That is, for given orthonormal bases in all local Hilbert spaces, there is a one-to-one correspondence between quantum states and matrix pencils. Throughout the remainder of the paper, we will denote matrix pencils either by $\mathcal{P}(R,S)$, $\mathcal{P}(\mu,\lambda)$, or simply by $\mathcal{P}$. SLOCC classes of $2\times m \times n$ states are then characterized by considering a normal form of the corresponding matrix pencils, which we will review below.

As the action of local invertible operators $A \otimes B \otimes C$ on quantum states is crucial for studying the SLOCC classes, we will now review how such an action transforms the corresponding matrix pencils. We will first consider the operator, $A$, acting on A's system and then consider the operators, $B,C$ acting on system $B$ and $C$ resp..
It is easy to verify that the operator $A=\begin{pmatrix} \alpha & \beta \\ \gamma & \delta \end{pmatrix}$, with $\det{A} \neq 0$, transforms the state,
$\ket{\psi} =  \ket{0}_A\ket{R}_{BC} + \ket{1}_A\ket{S}_{BC}$, to $(A\otimes \one \otimes \one)\ket{\psi}=\ket{0}_A (\alpha \ket{R}_{BC} + \beta \ket{S}_{BC}) + \ket{1}_A ( \gamma \ket{R}_{BC} + \delta \ket{S}_{BC})$. Hence, the corresponding pencil is transformed from $\mu R + \lambda S$ to $(\alpha \mu + \gamma \lambda) R + (\beta \mu + \delta \lambda) S$. Stated differently, the operator $A$, on the first system leads to a new matrix pencil, where the variables $\mu$ and $\lambda$ are transformed to the new variables $\hat{\mu}$ and $\hat{\lambda}$ via an invertible linear transformation. More precisely, $(\mu, \lambda)^T \rightarrow A^T (\mu, \lambda)^T=(\hat{\mu},\hat{\lambda})^T = (\alpha \mu + \gamma \lambda, \beta \mu + \delta \lambda)^T$, where $\alpha \delta - \beta \gamma \neq 0$. The operators on the second and third system transform the state $\ket{\psi}$ to  $(\identity \otimes B \otimes C) \ket{\psi} = \left[ \ket{0}_A (B R C^T \otimes \identity) + \ket{1}_A (B S C^T \otimes \identity) \right] 
\ket{\phi^+_n}_{BC}$. This corresponds to a transformation of the matrix pencil $\mu R + \lambda S$ to $B (\mu R + \lambda S) C^T$. From these observations it is evident that in order to identify the different SLOCC classes of states in $2\times m\times n$, a classification of matrix pencils is required. In the following subsections we first recall a normal form for matrix pencils and then review how it can be used to characterize the SLOCC classes. We will then use these results in order to analyze possible transformations from states in $2\times m\times n$ to states in a smaller dimensional Hilbert space.

\subsection{Normal form of matrix pencils}
\label{sec:pencils}
In this subsection we summarize some properties of matrix pencils that are used to characterize SLOCC classes. For an introduction to matrix pencils see \cite{Ga59} and \cite{ChMi10}. 
Let us begin by introducing some definitions and by recalling the normal form of matrix pencils. Here, and in the following we will denote by $GL_k$ the set of complex invertible $k\times k$ matrices. Moreover, we are using the following definition.
\begin{definition}
Two matrix pencils $\mathcal{P}(R,S)$  and $\mathcal{P}(R',S')$ of the same dimension, $m \times n$, are strictly equivalent to each other if there exist $B\in GL_m$ and $C\in GL_n$ such that $B\mathcal{P}(R,S) C^T= \mathcal{P}(R',S')$, i.e., $\forall \mu, \lambda \ B (\mu R + \lambda S) C^T = \mu R' + \lambda S'$.
\end{definition}
From the linearity of the pencil it follows that $\mathcal{P}(R,S)$  and $\mathcal{P}(R',S')$ are strictly equivalent iff there exist $B\in GL_m$ and $C\in GL_n$ such that $R'=BRC^T$ and $S'=BRC^T$. Hence, two states can be transformed into each other via operators applied by $B$ and $C$ iff the corresponding matrix pencil are strictly equivalent to each other.

Similarly to the well-known Jordan Normal Form (JNF) for matrices, a normal form for linear matrix pencils $\mathcal{P}(R,S)$ has been introduced, the so-called Kronecker Canonical Form (KCF) \cite{Kr90}. We denote a matrix pencil in KCF by $\mathcal{P}_{KCF}$ and discuss its structure in detail below.
A KCF of a matrix pencil $\mathcal{P}(R,S)$ has the following generalized block-diagonal form \footnote{Note that here a generalized block diagonal matrix denotes a matrix composed of rectangular blocks which are arranged diagonally. Moreover, note that $h$ or $g$ (or both) might be $0$.},
\begin{align}
\label{eq:kcf}
\mathcal{P}_{KCF}= \left\{0^{h\times g},  L_{\epsilon_1},\dots, L_{\epsilon_a}, L_{\nu_1}^T,\dots, L_{\nu_b}^T,J	\right\}.
\end{align}
The first block, $0^{h\times g}$, indicates that the first $h\geq 0$ rows as well as the first $g\geq 0$ columns of the matrix vanish. The other blocks will be defined below (see Eq. (\ref{eq:lblock}), Eq. (\ref{eq:jblock}), and Eq. (\ref{eq:mblock})). We restrict ourselves here to $2 \times m \times n$ states that are truly entangled in all dimensions\footnote{We assume everywhere that $n\leq 2m$.}, i.e., the local reduced density matrices, $\rho_A$, $\rho_B=R R^\dagger+ S S^\dagger, \rho_C=R^T R^\ast+ S^T S^\ast$, are of rank $2$, $m$, and $n$ respectively. Note that this implies that there does not exist a constant vector $\vec{v}$ which lies in the left or in the right nullspace of both $R$ and $S$. This is equivalent to the fact that there are no rows or columns identically 0 in the KCF of the pencil. Thus, we restrict ourselves to pencils for which $h=0$ and $g=0$, here and throughout the remainder of the article \footnote{Note, however, that the condition $h=g=0$ is not sufficient for the state to be truly entangled in all dimensions as Alice's 
system might still separate from Bob and Claire.}. 
A pencil is of rank $r$, if $r$ is the largest integer such that there exists an $r$-minor which is non-vanishing for some choice of $\mu$ and $\lambda$. Recall that an $r$-minor is given by the determinant of a matrix constructed by discarding all but $r$ rows and $r$ columns from the matrix pencil. We write $\operatorname{rk}\mathcal{P}(R,S) = r$ and in the following $r$ will denote the rank of the considered matrix pencil unless stated otherwise.

Before introducing the required definitions to explain the KCF, let us mention here that the blocks $L_{\epsilon}$ and $L_{\nu}^T$ are determined by the so-called \emph{minimal indices} of a matrix pencil, while $J$, which is itself a block-diagonal matrix, is constituted of blocks specified by the so-called \emph{elementary divisors} or, equivalently, through the so-called \emph{eigenvalues}, $x_i  \in \mathbb{C} \cup \{\infty\}$, and \emph{eigenvalue size signatures}, $s_i$, of a matrix pencil.

Let us now recall how the eigenvalues and corresponding size signatures as well as the minimal indices can be calculated and thus how the structure of the matrix pencil can be determined.
We define $D_k(\mu,\lambda)$ for $1 \leq k \leq r$, where $r = \operatorname{rk}\mathcal{P}(R,S)$, as the (polynomial) greatest common divisor of all the $k$-minors of a given matrix pencil $\mathcal{P}(R,S)$, i.e., $D_k(\mu,\lambda)=\gcd(\operatorname{minors}[\mathcal{P}(R,S),k])$. Furthermore, by convention $D_0 = 1$ and $D_k = 0$ for $k > r$.
The so-called \emph{invariant polynomials} $E_k(\mu,\lambda)$ are defined as $E_k(\mu,\lambda)=\frac{D_{k}(\mu,\lambda)}{D_{k-1}(\mu,\lambda)}$, where $1\leq k\leq r$. Since $D_{k-1}(\mu,\lambda)$ always divides $D_{k}(\mu,\lambda)$, $E_k(\mu,\lambda)$ is a homogeneous polynomial of $\mu$ and $\lambda$.
$D_r(\mu,\lambda)$ can be uniquely factorized as $D_r=\mu^{q-t} p_1 p_2\dots p_t$, where $p_i = \mu x_i + \lambda$, $t\leq q$, and $q \leq r$ is an integer that we will specify later. The $x_i$ occurring in $p_i$ are called the finite eigenvalues of the pencil, while the factor $\mu$, if present, corresponds to the eigenvalue $\infty$. From now on we will consider the set of distinct eigenvalues $\{x_i\}$. We will call the number of times that the term $x_i \mu + \lambda$ (the term $\mu$) is present in the factorization of $D_r$ the \emph{algebraic multiplicity} $e^i$ ($e^\mu$) of $x_i\neq \infty$ ($x_i = \infty$), respectively. The size signature corresponding to an eigenvalue $x_i \neq \infty$ is defined as the sequence of integers $s_i = (e_1^i, \ldots, e_r^i)$, where $e_j^i$ is the largest integer such that $(x_i \mu + \lambda)^{e_j^i}$ divides $E_j(\mu,\lambda)$. Likewise, $s_\mu$ is the signature of the eigenvalue $\infty$ and is given by $s_\mu = (e_1^\mu, \ldots, e_r^\mu)$, where $e_j^\mu$ is the 
largest integer such that $\mu^{e_j^\mu}$ divides $E_j(\mu,\lambda)$. The finite (infinite) elementary divisors corresponding to an eigenvalue $x_i$ are defined as $(x_i \mu + \lambda)^{e_j^i}$ in case $x_i \neq \infty$ ($\mu^{e_j^\mu}$ in case $x_i = \infty$), respectively.

Let us now also recall the notion of minimal indices of a matrix pencil.
To this end, let us consider the right null space of a matrix pencil $\mathcal{P}(R,S)$. More precisely, we consider the vectors $\vec{x'}_i$ of homogenous polynomials in $\lambda$ and $\mu$ with entries of coinciding degree, which fulfill
\begin{align}
(\mu R+\lambda S)\vec{x'}_i (\mu,  \lambda)=0.
\end{align}
We denote the degree of the entries of $\vec{x'}_i$ by $\epsilon'_i$, respectively. We can then write $\vec{x'}_i = \sum_{j=0}^{\epsilon'_i} \vec{x'}_{i j} \mu^{\epsilon'_i - j}\lambda^j$, where $\vec{x'}_{i j} \in \mathbb{C}^n$. Let furthermore $a$ denote the number of linearly independent vectors for which Eq. (\theequation) holds, i.e., the maximal number of $\vec{x'}_i$ obeying Eq. (\theequation) for which $q_1 \vec{x'}_1 + \ldots + q_k \vec{x'}_k=0$ with arbitrary polynomials $q_i$ has no non-trivial solution. By iteratively choosing linearly independent vectors $\vec{x}_i$ of minimal degree, $\epsilon_i$, and ordering them by degree in ascending order, one obtains the sequence of minimal indices, $(\epsilon_1, \ldots, \epsilon_a)$. 
Although the sequence $(\vec{x}_1, \ldots, \vec{x}_a)$ is not uniquely determined, their degrees $(\epsilon_1, \ldots, \epsilon_a)$ are \cite{Ga59}.
 Similarly, a sequence $(\nu_1, \ldots, \nu_b)$ can be obtained considering the equations $(\mu R^T+\lambda S^T)\vec{x}_i (\mu,  \lambda)=0$, i.e., considering the left null--space.

Using the definitions introduced above, we can now present the definition of the matrices in Eq. (\ref{eq:kcf}). The matrices denoted by $L_\epsilon$ are called right null-space blocks. $L_\epsilon$ has dimensions $\epsilon\times (\epsilon+1)$ and is defined as
\begin{equation}
\label{eq:lblock}
L_{\epsilon}=\begin{pmatrix}\lambda & \mu & 0 & \dots & 0\\
0 & \lambda & \mu & \ddots  & \vdots \\
\vdots & \ddots & \ddots & \ddots & 0 \\
0 & \hdots & 0 & \lambda & \mu
\end{pmatrix}.
\end{equation}
Conversely, $L^T_{\nu}$ is a left null-space block. It is given by the transpose of a right null-space block $L_\nu$. Note that $L_{\epsilon}$ has exactly one (linearly independent) vector in its right  null-space, which we denote by
\begin{align}
\label{eq:nullvector}
\vec{x}^\epsilon = \sum_{j=0}^{\epsilon} (-1)^j \mu^{\epsilon - j} \lambda^{j} \vec{e}_j,
\end{align}
where $\vec{e}_j$ denotes the $j$th unity vector. Hence, the sequences of minimal indices $(\epsilon_1, \ldots, \epsilon_a)$ and $(\nu_1, \ldots, \nu_b)$ determine the sizes of these blocks. As indicated in Eq. (\ref{eq:kcf}), there are $a$ $L$-blocks of respective sizes $\epsilon_i$ and $b$ $L^T$-blocks of respective sizes $\nu_i$.
The eigenvalues and their corresponding size signatures define the structure of the $J$ block as follows.
\begin{equation}
\label{eq:jblock}
J=\{M^{e_1^1}(x_1),\dots,  M^{e_r^l}(x_l), N^{e_1^\mu},N^{e_2^\mu},\dots, N^{e_r^\mu} \},
\end{equation}
where $l$ is the number of distinct finite eigenvalues and where $N^{e_j^\mu}$ and $M^{e_j^i}(x_i)$ are $e_j^\mu \times e_j^\mu$ and $e_j^i\times e_j^i$ matrices respectively given by
\begin{equation}
\label{eq:mblock}
N^{e_j^\mu}=\begin{pmatrix}\mu & \lambda & 0 & \dots & 0\\
0 & \mu & \lambda &  & \\
 &  & \ddots & \ddots & 0\\
 &   &  & \mu & \lambda\\
0 &  & \hdots & 0 & \mu
\end{pmatrix}\ \ \ \ \ \mbox{and }\quad M^{e_j^i}(x_i)=\begin{pmatrix}x_i \mu + \lambda & \mu & 0 & \dots & 0\\
0 & x_i \mu + \lambda & \mu &  &  \\
 &   & \ddots & \ddots & 0\\
 &   &  & x_i \mu + \lambda & \mu\\
0 &  & \hdots & 0 & x_i \mu + \lambda
\end{pmatrix}.
\end{equation}
Here, $e_j^\mu=0$ or $e_j^i=0$ indicates that the corresponding block is not present. Using the notation $J(x_i) = \bigoplus_j M^{e_j^i}(x_i)$ in case $x_i \neq \infty$ and $J(\infty) =  \bigoplus_j N^{e_j^\mu}$ we have $J = \bigoplus_i J(x_i)$.

Note that $J$ is a $q \times q$ matrix, i.e., the total size of $J$ is given by the degree of $D_r$ which can also be expressed as the sum of all signatures of the eigenvalues, $\sum_j(\sum_i e_j^i)+e_j^\mu$. The total dimensions for null-space blocks are then $(m-q) \times (m-q)$, which can be expressed using the minimal indices as $(b+ \sum_{i=1}^{a} \epsilon_i + \sum_{i=1}^{b} \nu_i) \times (a+ \sum_{i=1}^{a} \epsilon_i + \sum_{i=1}^{b} \nu_i)$ \footnote{Recall that $h=g=0$ in Eq. (\ref{eq:kcf}).}.
After properly defining an ordering of the blocks in a final step, a unique normal form for matrix pencils, the Kronecker Canonical Form, is obtained. It has been shown in Ref. \cite{Kr90} that any matrix pencil is strictly equivalent to its KCF. Hence, 
strict equivalence of two matrix pencils can then be decided by comparing the respective KCFs, as stated by the following lemma, which is proven in Ref. \cite{Kr90}.
\begin{lemma}\label{Kronecker}
Two matrix pencils are strictly equivalent to each other iff they have the same Kronecker Canonical Form.
\end{lemma}

Recall that the KCF is uniquely determined by the eigenvalues, their size signatures and the minimal indices of a matrix pencil. The statement could thus also be formulated the following way. Two matrix pencils are strictly equivalent iff their eigenvalues with corresponding size signatures and the minimal indices coincide.

Let us at this place summarize the procedure to determine the KCF of an arbitrary given matrix pencil. First, one calculates the greatest common divisors of the $k$-minors of the matrix pencil, $D_k(\mu,\lambda)=\gcd(\operatorname{minors}(\mathcal{P}(R,S),k))$. One then calculates the invariant polynomials $E_k(\mu,\lambda)=\frac{D_{k}(\mu,\lambda)}{D_{k-1}(\mu,\lambda)}$. One determines the eigenvalues and their size signatures by considering the factorizations of the invariant polynomials as explained above. In the next step the minimal indices of the matrix pencil are determined. To this end, one determines a homogenous polynomial vector of minimal degree, $\vec{x}_1$, which fulfills $\mathcal{P} \vec{x_1} = 0$. The degree of $\vec{x}_1$ is the first minimal index $\epsilon_1$. One iteratively determines $\vec{x}_i$ of minimal degree which is linearly independent (in the sense explained above) from $\vec{x}_1, \ldots, \vec{x}_{i-1}$ and fulfills $\mathcal{P} \vec{x_i} = 0$ until no such vector can be found 
any more. This procedure uniquely yields the sequence of minimal indices $(\epsilon_1, \ldots, \epsilon_a)$. Considering $\mathcal{P}^T \vec{x_i} = 0$ one similarly obtains the sequence of left minimal indices $(\nu_1, \ldots, \nu_b)$.

\subsection{Examples of computing the KCF}
\label{sec:kcfcomputation}
In this subsection, we first present the basic examples of matrix pencils corresponding to the well known W and GHZ states in $\mathbb{C}^2 \otimes \mathbb{C}^2 \otimes \mathbb{C}^2$. Then we show how the KCF can be computed for two other examples.

The matrix pencils $\mathcal{P}_W$ and  $\mathcal{P}_{GHZ}$ corresponding to the three-qubit states $\ket{W}$ and $\ket{GHZ}$ given in Eq. (\ref{W}) and (\ref{GHZ}) are given by
\begin{equation}
\mathcal{P}_{W}=\begin{pmatrix}\lambda & \mu\\
\mu & 0
\end{pmatrix}\quad\quad\mbox{and\quad\quad}\mathcal{P}_{GHZ}=\begin{pmatrix}\mu & 0\\
0 & \lambda
\end{pmatrix}.
\end{equation}
It can easily be verified that through column and/or row permutations, these matrix pencils can be brought into the respective KCF 
\begin{equation}
\mathcal{P}_{W,KCF}=  N^{2} =\begin{pmatrix}\mu & \lambda\\
0 & \mu
\end{pmatrix}\quad\quad\mbox{and\quad\quad}\mathcal{P}_{GHZ,KCF}= M^{1}(0) \oplus N^{1} = \begin{pmatrix}\lambda & 0\\
0 & \mu
\end{pmatrix}.
\end{equation}

Let us now show how to compute the KCF of the following matrix pencil
\begin{align}
\mathcal{P} = \begin{pmatrix}
\lambda & \mu & 0 & 0 & \lambda \\
\lambda & \lambda & \mu & \lambda + \mu & 0\\
3 \mu & -\lambda & -\mu & 2 \mu & 0\\
\mu & 0 & 0 & 0 & 2 \mu\\
\end{pmatrix}.
\end{align}
First, we calculate the $4$-minors of this matrix pencil. There are 5 of them which are obtained by calculating the determinant of the submatrix obtained by deleting one of the five columns of $\mathcal{P}$. The minors are
$- \mu^3 (3 \mu + \lambda)$,
$\mu (3 \mu + \lambda) (i \sqrt{2} \mu + \lambda) (-i \sqrt{2} \mu + \lambda)$,
$\mu^2 \lambda (3 \mu + \lambda)$, and two of the minors equal 
$2 \mu^3 (3 \mu + \lambda)$.
As their greatest common divisor is $\mu (3 \mu + \lambda)$, $D_4 = \mu (3 \mu + \lambda)$. Moreover, the rank of the pencil is 4. To determine $D_3$, we observe that the 3-minor obtained by deleting the first and fourth column and the third row of $\mathcal{P}$ equals $2 \mu^3$, and the 3-minor obtained by deleting the third and fourth column and the fourth row of $\mathcal{P}$ equals $- \lambda^2 (3 \mu + \lambda)$. Hence, their greatest common divisor is 1 and therefore $D_3=1$. From that it follows, as $D_{k-1}$ divides $D_k$, that $D_2=D_1=1$. We can now calculate the invariant polynomials $E_k = \frac{D_k}{D_{k-1}}$ and obtain $E_4 =  \mu (3 \mu + \lambda)$ and $E_3 = E_2 = E_1 = 1$. Hence, the elementary divisors are $3 \mu + \lambda$ and $\mu$. Thus, we know that the pencil has two distinct eigenvalues, $x_1 = 3$ corresponding to the divisor $(3 \mu + \lambda)$ and $x_2 = \infty$ corresponding to the divisor $\mu$. Hence a block $J = \begin{pmatrix}
3 \mu + \lambda & 0 \\
0 & \mu
\end{pmatrix}$ is present in the KCF of the pencil. Let us now determine the nullspace blocks present in the KCF of the pencil. One possibility to do so would be to determine the minimal indices of the pencil $\mathcal{P}$ by the procedure described above. However, in this example the nullspace structure can be derived using dimension arguments. From the dimensionality of the pencil it follows that the nullspace blocks form a block of size $2 \times 3$, as a $J$ block of size $2 \times 2$ is present. Due to the fact that the rank of the pencil is 4 we have that $h=0$. Now, as the only way to distribute $L_{\epsilon_i}$ and $L^T_{\nu_i}$ blocks in a $2 \times 3$ matrix is $L_2$, we have that the KCF of the pencil is
\begin{align}
\mathcal{P}_{KCF} = \begin{pmatrix}
\lambda & \mu & . & . & . \\
. & \lambda & \mu & . & .\\
. & . & . & 3 \mu + \lambda & .\\
. & . & . & . & \mu\\
\end{pmatrix}.
\end{align}

Let us now consider as a second example the following $m \times m$ matrix pencil, which will be relevant later
\begin{equation}
 \mathcal{P}=\begin{pmatrix}
\lambda & \mu & \cdot & \cdot & \cdot \\
\cdot & \ddots & \ddots & \cdot & \cdot\\
\cdot & \cdot & \lambda & \mu & \cdot\\
\cdot & \cdot & \cdot & \lambda & \mu\\
-a_0 \mu & \cdots & \cdot & -a_{m-2} \mu & -a_{m-1} \mu + \lambda &
\end{pmatrix}.
\end{equation}
To determine the pencil's KCF, we first evaluate the greatest common divisor of its $k$-minors, $D_k$, for $1\leq k\leq m$. First note that $D_k=1$ for all $1 \leq k \leq m-1$. This can be seen as follows. Among the $(m-1)$-minors, there is one minor equal to $\lambda^{m-1}$ and another one equal to $\mu^{m-1}$. Hence, their greatest common divisor equals 1 and therefore $D_{m-1} = 1$ which implies that $D_{k} = 1$ for all $1 \leq k \leq m-1$. The $m$-minor, however, equals the determinant of $\mathcal{P}$ which can be expressed as
\begin{equation}
\label{eq:dmwitha}
D_m(\mu,\lambda)= \lambda^{m}+\sum^{m-1}_{i=0}(-1)^{m-i}a_{i}\mu^{m-i}\lambda^{i}.
\end{equation}
As $D_m$ is of degree $m$, there must exist $m$ (not necessarily distinct) eigenvalues. Recall that this implies that the size of the $J$ block equals $q \times q$ with $q = m$. Hence, the KCF of the matrix pencil contains no nullspace blocks, only the $J$ block is present. Moreover, as $D_m(\mu=0, \lambda) \neq 0$, all the eigenvalues are finite. Denoting the (here possibly non-distinct) eigenvalues as $x_1, \ldots, x_m$, $D_m$ can, hence, be expressed as
 \begin{equation} \label{eq:dmwithx}
 D_m=(\mu x_1+\lambda)(\mu x_2+\lambda)\dots (\mu x_m+\lambda).
 \end{equation}
The coefficients, $a_0, a_1,\dots, a_{m-1}$, are thus given by
 \begin{equation}\label{coeff}
a_k = (-1)^{m-k} \bigg(\sum_{\underset{\sum_{j=1}^{m}i_{j}=(m-k)}{i_{1},i_{2},\dots,i_{m}\in\{0,1\}}}x^{i_1}_1x^{i_2}_2\dots x^{i_m}_m \bigg),
 \end{equation}
 which can be verified by comparing the coefficients of the polynomials in $\mu$ and $\lambda$ in Eq. (\ref{eq:dmwitha}) and Eq. (\ref{eq:dmwithx}).

Let us denote by $\{\tilde{x}_i\}_i$ the set of distinct eigenvalues. Considering that $D_{m-1}=1$ and therefore $E_k = 1$ for all $k \leq m-1$, the size signature corresponding to an eigenvalue $\tilde{x}_i$ is $s_i = (0, \ldots, 0, m_i)$, where $m_i$ is the algebraic multiplicity of eigenvalue $\tilde{x}_i$. Hence the KCF of the matrix pencil is of the form
\begin{align}
\label{equ:companionkcf}
\mathcal{P}_{KCF}=\begin{pmatrix}
\tilde{x}_1 \mu +  \lambda & \mu & \cdot & \cdot & \cdot & \cdot & \cdot\\
\cdot &\ddots & \ddots & \cdot & \cdot & \cdot & \cdot\\
\cdot & \cdot & \tilde{x}_1 \mu + \lambda & \mu & \cdot& \cdot & \cdot\\
\cdot & \cdot & \cdot & \tilde{x}_1 \mu + \lambda & \cdot & \cdot & \cdot\\
\cdot & \cdot & \cdot & \cdot & \tilde{x}_2 \mu + \lambda & \mu & \cdot\\
\cdot & \cdot  & \cdot & \cdot & \cdot & \ddots & \ddots
\end{pmatrix} =  \bigoplus_{i} J(\tilde{x}_i) =  \bigoplus_{i} M^{m_i}(\tilde{x}_i),
\end{align}
where the size of $J(\tilde{x}_i)$ is $m_i \times m_i$.

In case the eigenvalues $x_1, x_2, \dots, x_m \in \mathbb{C}$ are all distinct, the KCF is diagonal and one can easily determine the matrices $B$ and $C^T$ bringing the given matrix pencil to its KCF, i.e., $B  \mathcal{P} C^T = \mathcal{P}_{KCF}$. In order to see that, note that the matrix pencil $\mathcal{P}$ is of the form $\lambda \identity + \mu \mathcal{C}$, where $\mathcal{C}$ is the transpose of a so-called  companion matrix \cite{HoJo13}. We can hence use the Vandermonde matrix
\begin{equation}
\label{EqVandermonde}
V=\begin{pmatrix}
1 & x_{1} & x_1^2 & \ldots & x_1^{m-1}\\
1 & x_{2} & x_2^2 & \ldots & x_2^{m-1}\\
\vdots & \vdots & \vdots & \ddots & \vdots \\
1 & x_{m} & x_m^2 & \ldots & x_m^{m-1}\\
\end{pmatrix}
\end{equation}
to diagonalize $\mathcal{C}$ as $V \mathcal{C}^T V^{-1} = \operatorname{diag}(x_1, x_2, \ldots, x_m)$, where $\operatorname{diag}(x_1, x_2, \ldots, x_m)$ denotes a diagonal matrix with the entries $x_1, x_2, \ldots, x_m$ \cite{HoJo13}. Hence, the matrices bringing $\mathcal{P}$ to its KCF are $B={V^{-1}}^{T}$ and $C^T =  V^T$. 
Note that if the eigenvalues, $x_i$, are degenerate then $\mathcal{C}$ is not diagonalizable. In this case a similarity transformation $\tilde{V}$ can be used, which transforms $\mathcal{C}$ to a block diagonal matrix, where each block is a Jordan block of size given by the algebraic multiplicity of the eigenvalues. Hence, the matrix pencil, $\lambda \identity + \mu \mathcal{C}$, can be transformed with this similarity transformation $\tilde{V}$ to its KCF given in Eq. (\ref{equ:companionkcf}).

 \subsection{Characterization of SLOCC classes in $2\times m\times n$}
 \label{sec:alicedetails}

 Lemma \ref{Kronecker} implies that two states can be transformed into each other via some invertible operator of the form $\identity \otimes B \otimes C$ iff the corresponding matrix pencils have the same KCF. However, the effect of the invertible operator, $A$, applied on the first system (the qubit), remains to be taken into account. As
 explained above it transforms the variables of the matrix pencil $(\mu, \lambda)$ to some new variables $(\hat{\mu},\hat{\lambda})^T=A^T(\mu, \lambda)^T$. Let us now review how the matrix pencil, or more precisely, the KCF of a matrix pencil, changes under the action of Alice.

 It has been shown that the minimal indices of a matrix pencil cannot be altered by $A$ \cite{ChMi10, DeEd95}. However, the eigenvalues, $x_i$, of a matrix pencil can indeed change, while the size signatures $s_i$ remain unchanged. Denoting by $A=\begin{pmatrix} \alpha & \beta \\ \gamma & \delta \end{pmatrix}$, with $\det{A} \neq 0$ the operator applied to the first system, the finite eigenvalues $x_i \neq \infty$ transform to \cite{ChMi10}
  \begin{align}
   x_i \rightarrow \left\{\begin{array}{lr}
        \frac{\alpha x_i + \beta}{\gamma x_i + \delta}, & \text{if } \gamma x_i + \delta \neq 0\\
       \infty, & \text{if } \gamma x_i + \delta = 0\\
        \end{array}\right.,
 \end{align}
 while the infinite eigenvalue changes as
  \begin{align}
   \infty \rightarrow \left\{\begin{array}{lr}
        \frac{\alpha}{\gamma}, & \text{if } \gamma \neq 0\\
       \infty, & \text{if } \gamma = 0\\
        \end{array}\right..
 \end{align}
As can be seen from here, $A$'s action can always be used to bring a state to a form with no infinite eigenvalues \cite{ChMi10,Ga59}. Hence, in order to study SLOCC classes, only states with finite eigenvalues need to be considered. With this we are ready to formulate the theorem which relates matrix pencils to the SLOCC equivalence classes of $2 \times m \times n$-dimensional states.
 \begin{theorem}[\cite{ChMi10}]
 \label{theo:fractional}
Two  $2 \times m \times n$-dimensional pure states $\ket{\psi}$ and $\ket{\phi}$ for which their corresponding matrix pencils have only finite eigenvalues are SLOCC equivalent iff they have the same minimal indices, matching eigenvalue size signatures, and the eigenvalues  $\{x_i\}$ and $\{x_i'\}$, respectively, are related by a linear fractional transformation
\begin{equation}
\frac{\alpha x_i+ \gamma}{\beta x_i+\delta}=x_i', \text{ for some } \alpha, \beta, \gamma, \delta \in \mathbb{C}, \text{ where } \alpha \delta - \beta \gamma \neq 0.
\end{equation}
\end{theorem}
 Note that here the invertible operation $A=\begin{pmatrix} \alpha & \beta \\ \gamma & \delta \end{pmatrix}$ is applied by Alice.
As noted in \cite{ChMi10}, it is always possible to uniquely relate two triplets of distinct eigenvalues $\{x_1,x_2,x_3\}$ and $\{x_1',x_2',x_3'\}$ with a linear fractional transformation \cite{Brown}. In particular, this implies that a state with a corresponding matrix pencil that has only three distinct eigenvalues is SLOCC equivalent to any other state with a corresponding matrix pencil that has matching minimal indices, eigenvalue size signatures, but three arbitrary (distinct) eigenvalues. For more than three distinct eigenvalues this does not hold in general.

\subsection{Local transformations form larger- to smaller-dimensional Hilbert space}

In the previous two subsections we reviewed how the SLOCC classes in $2 \times m\times n$ can be characterized via matrix pencils. Here, we go one step beyond that and recall the necessary and sufficient condition for the existence of a local (non--invertible) transformation from a state $\ket{\psi}$ with local ranks $2 \times m \times n$ to a state $\ket{\phi}$ with local ranks $2 \times m \times k$ where $k<n$. In order to do so, we consider now a (non-invertible) SLOCC transformation of a state with local rank $2 \times m \times n$ to a state with local rank $2 \times m \times (n-1)$. The general case can be deduced from that by iterating the process. Note that such a transformation requires a non-invertible SLOCC operation performed by $C$. The necessary and sufficient conditions for the existence of this operation are stated in the following theorem which is presented in \cite{ChMi10}.
Here and in the following, we will denote the matrix pencil that is associated to a state $\ket{\psi}$ by $\mathcal{P}_\psi$.

\begin{theorem}[\cite{ChMi10}]\label{theorem:traforules} Let $\ket{\psi}$ and $\ket{\phi}$ be states with local ranks $2\times m\times n$ and $2\times m\times(n-1)$  and let $c_1, c_2, \dots, c_n$ denote the columns of the pencil $\mathcal{P}_{\psi}{(\mu, \lambda)}$. Then $\ket{\psi}$ can be mapped to $\ket{\phi}$ via some non-invertible SLOCC operators iff for some $1\leq i \leq n$, there exist constants $a_1, a_{i-1},a_{i+1},\dots, a_n$  and some invertible linear transformation $(\mu ,\lambda)\mapsto (\hat{\mu},\hat{\lambda})$ such that the pencil $\mathcal{P}_{\psi_i}{(\hat{\mu}, \hat{\lambda})}=[c_1+a_1c_i,\dots,  c_{i-1}+a_{i-1}c_i, c_{i+1}+a_{i+1}c_i, \dots, c_n+a_nc_i ]$ is strictly equivalent to $\mathcal{P}_{\phi}(\mu, \lambda)$.
\end{theorem}

Let us remark here, that in the transformation from $\ket{\psi}$ to $\ket{\phi}$, the pencil $[c_1+a_1c_i,\dots,  c_{i-1}+a_{i-1}c_i, c_{i+1}+a_{i+1}c_i, \dots, c_n+a_nc_i ]$ is obtained by the third party, Claire, applying an operator $C$ that is given via the $n-1 \times n$ matrix 
\begin{align}
\label{eq:clairemap}
C=P_{n-1}(\one+ \ket{\phi}\bra{i}),
\end{align}
where $P_{n-1}$ denotes the projector onto the $n-1$ dimensional subspace spanned by all standard basis vectors but $\vec{e}_i$ and $\ket{\phi}=\sum_{j\neq i} a_j \ket{j}$. The intuition behind the theorem is that any matrix $\tilde{C}$ can be brought to its so-called \emph{reduced row echelon form} $C$ using invertible row operations, i.e., matrices that act on the left, only \cite{HoJo12}.

Similarly, transformations from states of local rank $2 \times m \times n$ to states with local rank $2 \times (m-1) \times n$ can be achieved by some non-invertible operation $B$ performed by Bob. Non-invertible operators applied by Alice are not considered as they would always leave Alice not entangled with the rest of the parties.

Let us illustrate this result with a simple example. Here and in the following we will say, similarly to Ref. \cite{ChMi10}, that a state has a State KCF (SKCF) if its corresponding matrix pencil has the KCF \footnote{Note that for simplicity, we do not fix the eigenvalues or the order of them, which implies that the SKCF is not unique here. Contrary to that a unique SKCF has been defined in Ref. \cite{ChMi10}.}. 
Let us consider the two states in SKCF
$\ket{\psi} = \ket{001} + \ket{012} + \ket{100} + \ket{111} + \ket{123} \in \mathbb{C}^2\otimes \mathbb{C}^3\otimes \mathbb{C}^4$ and
$\ket{\phi} = \ket{001} + \ket{012} + \ket{100} + \ket{111} + \ket{122} \in \mathbb{C}^2\otimes \mathbb{C}^3\otimes \mathbb{C}^3$.
Their corresponding matrix pencils are given by
\begin{align}
\mathcal{P}_\psi =  L_2 \oplus M^1(0) = \begin{pmatrix}
\lambda & \mu & \cdot & \cdot \\
\cdot & \lambda & \mu & \cdot \\
\cdot & \cdot & \cdot & \lambda \\
\end{pmatrix}
\text{ and }
\mathcal{P}_\phi =  M^3(0) = \begin{pmatrix}
\lambda & \mu & \cdot  \\
\cdot & \lambda & \mu  \\
\cdot & \cdot & \lambda \\
\end{pmatrix}.
\end{align}
It is easy to see that a matrix pencil that is strictly equivalent to $\mathcal{P}_\phi$ can be obtained (in particular, $\mathcal{P}_\phi$ can be obtained) by adding the last column of $\mathcal{P}_\psi$ to the third column of $\mathcal{P}_\psi$ and discarding the last column afterwards. According to Theorem \ref{theorem:traforules} this implies that $\ket{\psi}$ can be transformed to $\ket{\phi}$. The operator $C$ given in Eq. (\ref{eq:clairemap}), which performs the described action, reads
\begin{align}
C^T = \begin{pmatrix}
1 & . & . \\
. & 1 & . \\
. & . & 1 \\
a_1 & a_2 & a_3 \\
\end{pmatrix},
\end{align}
where $a_1 = a_2 = 0$ and $a_3 = 1$. As can be easily seen we have $\mathcal{P}_\psi C^T = \mathcal{P}_\phi$ and thus $\identity \otimes \identity \otimes C \ket{\psi} = \ket{\phi}$. Note that in general, operators $A$ and $B$ not equal to $\identity$ may be required. Note also that the operators are not unique.

\section{Characterization of generic states in $2\times m \times n$}

In this section we characterize a generic set of states in $2\times m \times n$. In order to do so, we first introduce some lemmata concerning a necessary and sufficient condition for a matrix pencil to be a direct sum of right null-space blocks $L_{\epsilon_i}$ only, and a sufficient condition on the matrix pencil corresponding to a state such that the operation applied by the first party can be inverted by a transformation of party $2$ and $3$. We will use these lemmata for both, the characterization of generic sets of states and to prove that a generic state can always be transformed into any state in a full measure set of states of smaller dimension. The latter is the main result of the subsequent section.
Using then the characterization of generic matrix pencils presented in Ref. \cite{DeEd95}, we show that the union of SLOCC classes of states corresponding to a generic set of matrix pencils is of full measure (see Theorem \ref{theo:genstates}). Interestingly, it turns out that this generic set of states (similarly to the generic set of matrix pencils) is characterized by $m-3$ parameters in case $m=n$ and has no free parameter in case $m\neq n$. 

Whereas the proof of Theorem \ref{theo:genstates} is presented in Appendix A, we present the proofs of some lemmata in the main text. The reason for that is that the proofs are illuminating for readers not that familiar with matrix pencils.

In the following lemma a characterization of matrix pencils being a direct sum of right null-space blocks is presented.
As stated before, we consider here only full ranked states in $2 \times m \times n$, i.e., states whose local rank is maximal. In particular, for all the KCF of the corresponding matrix pencils $h=g=0$ in Eq. (\ref{eq:kcf}). Recall that for any $m \times n$ matrix pencil, $D_m=1$ implies that the rank of the pencil is $m$.

\begin{lemma}\label{lemma:dmequalone} Let $\mathcal{P}(\mu, \lambda)$ denote a $m \times n$ matrix pencil, where $n>m$. Then the following two statements are equivalent.
\begin{enumerate}[(i)]
\item  $D_m=1$, 
\item $\mathcal{P}(\mu, \lambda)$ is strictly equivalent to a direct sum of right nullspace blocks $L_{\epsilon_i}$ only. 
\end{enumerate}
 In particular, for an $m\times (m+1)$ matrix pencil $\mathcal{P}(\mu, \lambda)$ we have $D_m=1$ iff $\mathcal{P}(\mu, \lambda)$ is strictly equivalent to $L_m$.
\end{lemma}
\begin{proof}
If $\mathcal{P}(\mu,\lambda)$ is strictly equivalent to a direct sum of $L_{\epsilon_i}$ blocks, then it is straightforward to see that $D_m=1$, as one of the $m$-minors equals $\mu^m$ and another one equals $\lambda^m$ and their greatest common divisor is therefore 1. Let us now prove that the converse also holds. First note, that in case $D_m=1$ there is no $J$ block present in the pencil's KCF.
To see this, note that the rank of the matrix pencil is $m$, which implies that $D_m = \mu^{e^\mu} \prod_{i,x_i\neq \infty} (x_i \mu + \lambda)^{e^i}$, which can only equal 1 if there exists no eigenvalue and hence no $J$ block \footnote{One can also see this by considering the invariant polynomials. Let us assume that a $J$ block is present, which implies that there exist eigenvalues, either finite or infinite. Recall that if this is the case, then there exists an invariant polynomial $E_k = \frac{D_k}{D_{k-1}} = \mu^{e^\mu_{k}} \prod_{i,x_i\neq \infty} (x_i \mu + \lambda)^{e_k^i} \neq 1$. As $D_k$ divides $D_m$ for all $k \leq m$ and we have at least one $k$ for which $D_k\neq 1$, we also have that $D_m \neq 1$ which contradicts the assumption that $D_m = 1$. This implies that there is no $J$ block present.}. We hence have a direct sum of $L_{\epsilon_i}$ and $L_{\nu_i}^T$ blocks only. Due to dimensionality reasons, the number of $L_{\epsilon_i}$ blocks equals $n-m$ plus the number of $L_{\nu_i}^T$ blocks. 
We now have to show, however, that there cannot be any left nullspace blocks. Let us assume the contrary, i.e., there is at least one $L_{\nu}^T$ block present. We will show that in this case all $m$-minors vanish implying $D_m=0$, which contradicts $D_m=1$. Without loss of generality, in particular having the same set of $m$-minors, we can write the matrix pencil as  
\begin{equation}
\begin{tikzpicture}[baseline=(current  bounding  box.center),
style1/.style={
  matrix of math nodes,
  every node/.append style={text width=#1,align=center,minimum height=2ex},
  nodes in empty cells,
  left delimiter=(,
  right delimiter=),
  },
]
\matrix[style1=0.2cm] (1mat)
{
  \lambda & & & \vphantom{\lambda} & & & & & & & & & \\
  \mu & \lambda & & & & & & & & & & & \\
  & \mu &  & & & & & & & & & & \\
  & & & \lambda & & & & & & & & & \\
  \vphantom{\mu}& & & \mu & & & & & & & & & \\
  & & & & & & & & & & & & \\
  & & & & & & & & & & & & \\
  & & & & & & & & & & & & \\
  & & & & & & & & & & & & \\
  & & & & & & & & & & & & \\
  & & & & & & & & & & & & \\
};

\node 
  at ([xshift=0pt,yshift=0pt]1mat-3-3) {$\ddots$}; 
\node
  at ([xshift=0pt,yshift=0pt]1mat-4-3) {$\ddots$}; 
\node[font=\huge] 
  at ([xshift=0pt,yshift=0pt]1mat-9-9) {$\cdots$}; 

\draw[solid]
  (1mat-1-1.north west) -- (1mat-1-4.north east);
\draw[solid]
  (1mat-5-1.south west) -- (1mat-5-4.south east);
\draw[solid]
  (1mat-1-1.north west) -- (1mat-5-1.south west);
\draw[solid]
  (1mat-1-4.north east) -- (1mat-5-4.south east);
  
\draw[solid]
  (1mat-6-5.north west) -- (1mat-6-13.north east);
\draw[solid]
  (1mat-11-5.south west) -- (1mat-11-13.south east);
\draw[solid]
  (1mat-6-5.north west) -- (1mat-11-5.south west);
\draw[solid]
  (1mat-6-13.north east) -- (1mat-11-13.south east);

\draw[decoration={brace,raise=12pt},decorate]
  (1mat-1-13.north east) -- 
  node[right=15pt] {$\nu + 1$} 
  (1mat-5-13.south east);
\draw[decoration={brace,raise=12pt},decorate]
  (1mat-6-13.north east) -- 
  node[right=15pt] {$m - (\nu + 1)$} 
  (1mat-11-13.south east);  
\draw[decoration={brace,mirror,raise=5pt},decorate]
  (1mat-11-1.south west) -- 
  node[below=7pt] {$\nu$} 
  (1mat-11-4.south east);  
\draw[decoration={brace,mirror,raise=5pt},decorate]
  (1mat-11-5.south west) -- 
  node[below=7pt] {$n - \nu$} 
  (1mat-11-13.south east);

\node at ([xshift=-20pt,yshift=-1.2pt]1mat.west) {$\mathcal{P} =$};  
\end{tikzpicture}.
\end{equation}

For computing any of the $m$-minors, there are at most $\nu$ vectors which are non-vanishing in the first $\nu + 1$ components. Hence, this $(\nu+1)$-dimensional subspace can never be spanned by those vectors and hence, the minor vanishes. This completes the proof. The statement about $m \times (m+1)$ matrix pencils follows immediately from the fact that if $n=m+1$ there exists no other direct sum of right null space blocks which amounts to the required dimension. 
\end{proof}

Let us now show that states corresponding to matrix pencils, which are direct sums of (right and left) null-space blocks only are very special and somehow resemble bipartite states. In fact, as shown in the following lemma, for those states the transformation accomplished by applying an operator on the first system can also be achieved by party $2$ and $3$. 
\begin{lemma}
\label{lemma:undoalice}
 If a matrix pencil $\mathcal{P}_\psi(\mu, \lambda)$ corresponding to a state $\ket{\psi}$ consists only of null-space blocks, i.e., $\mathcal{P}_{\psi} = \left\{L_{\epsilon_1},\dots L_{\epsilon_a}, L_{\nu_1}^T,\dots L_{\nu_b}^T \right\}$, then any invertible action of Alice, $A$, can be undone by Bob and Charlie with some invertible matrices $B$ and $C$.
 \end{lemma}

Hence, if the premises of the lemma are satisfied then for any operators $A\in GL_2$ there exist operators $B\in GL_m, C\in GL_n$ such that $A\otimes \one \otimes \one \ket{\psi}=
\one \otimes B \otimes C \ket{\psi}$. This property resembles a property of bipartite states as for any bipartite state, $\ket{\phi}$, we have that for any operators $A$ there exist operators $B$ such that $A\otimes \one \ket{\phi}=\one \otimes B \ket{\phi}$.
\begin{proof}[Proof of Lemma \ref{lemma:undoalice}] To prove this lemma, we make use of Lemma 3 in Ref. \cite{ChMi10}, which states that the minimal indices of a matrix pencil are invariant under the action of Alice. However, the KCF of a pencil that only contains null-space blocks is completely determined by the minimal indices. This implies that for any action of Alice, there exist operators $B$ and $C$ for Bob and Claire which bring the matrix pencil corresponding to the state back to its KCF. This completes the proof.
 \end{proof}

 Let us here consider as a simple example the state $\ket{\psi} = \ket{001} + \ket{013} + \ket{024} + \ket{100} + \ket{112} + \ket{123}  \in \mathbb{C}^2 \otimes \mathbb{C}^3 \otimes \mathbb{C}^5$ whose corresponding matrix pencil is given by
 \begin{align}
 	\mathcal{P}_\psi = L_1 \oplus L_2 =
 	\begin{pmatrix}
\lambda & \mu & \cdot & \cdot & \cdot \\
\cdot & \cdot & \lambda & \mu & \cdot \\
\cdot & \cdot & \cdot & \lambda & \mu \\
\end{pmatrix}.
 \end{align}
Alice now applies for instance the operator $A = \begin{pmatrix}
1 & 1\\
1 & 0\\
\end{pmatrix}$ transforming the state to $\ket{\psi'} = (A \otimes \identity \otimes \identity) \ket{\psi} = \ket{000} + \ket{001} + \ket{012} + \ket{013} + \ket{023} + \ket{024} + \ket{101} + \ket{113} + \ket{124}$ with a corresponding matrix pencil
 \begin{align}
 	\mathcal{P}_{\psi'} =
 	\begin{pmatrix}
\mu & \mu + \lambda & \cdot & \cdot & \cdot \\
\cdot & \cdot & \mu & \mu + \lambda & \cdot \\
\cdot & \cdot & \cdot & \mu & \mu + \lambda \\
\end{pmatrix}.
 \end{align}
One can easily verify that the operations
\begin{align}
B = \begin{pmatrix}
 1 & 0 & 0 \\
 0 & 0 & 1 \\
 0 & 1 & 1 \\
\end{pmatrix} \text{ and }
C = \begin{pmatrix}
 -1 & 1 & 0 & 0 & 0 \\
 1 & 0 & 0 & 0 & 0 \\
 0 & 0 & 1 & -1 & 1 \\
 0 & 0 & -2 & 1 & 0 \\
 0 & 0 & 1 & 0 & 0 \\
\end{pmatrix}
\end{align}
transform the matrix pencil back to its original form, i.e., $B \mathcal{P}_{\psi'} C^T = \mathcal{P}_{\psi}$. Stated differently, $A \otimes B \otimes C$ is a symmetry of the state $\ket{\psi}$.

In Ref. \cite{DeEd95} a generic set of matrix pencils has been characterized. That is a set of matrix pencils $G=\{\mathcal{P}_i\}$ has been identified with the property that the union of the orbits of matrix pencils within this set is of full measure. That is, $ \bigcup_i {\cal O}(\mathcal{P}_i)$ is of full measure, where ${\cal O}(\mathcal{P}_i)=\{B\mathcal{P}_iC^T, \mbox{ with } B\in GL_m,C\in GL_n\}$. In $2\times m\times m$ it is easy to verify that the set of matrix pencils with distinct eigenvalues is generic \cite{Jo09}. In $2\times m\times n$, with $d=n-m\geq 1$, however, there is only one matrix pencil, whose orbit is generic. It is given by a direct sum of (at most two different) null-space blocks. More precisely, the following theorem has been proven in Ref. \cite{DeEd95}.

\begin{theorem} \label{ThGenPencils} A generic matrix pencil of dimension $m\times m$ corresponds to the matrix pencil with $m$ distinct divisors. If $d=n-m\geq 1$, the generic matrix pencil is given by $\mathcal{P}(R,S)=(d-(m\ mod\ d))L_{\left\lfloor m/d \right\rfloor}\oplus (m\ mod\ d)L_{\left\lceil m/d\right\rceil }$, where $\lfloor . \rfloor$ ($\lceil . \rceil$) denotes the floor (ceiling) function, respectively.
\end{theorem}

Let us mention here that in order to prove this result the codimension of the orbit, $\mathcal{O}(\mathcal{P})$ has been computed. See \cite{DeEd95} and references therein. It is defined as the difference between the dimension of the whole space, i.e., $2 m n$ (counting complex dimensions), and the dimension of the 
orbit, $\operatorname{dim}\left\{\mathcal{O}(\mathcal{P})\right\}$.
It is evident that the codimension is minimal for $n=m$ if there are only distinct divisors (as for such a pencil the dimension of the symmetries, $S_\mathcal{P}=\{(B,C): B\mathcal{P}C^T=\mathcal{P}\}$, is the smallest). However, in case $d=n-m\geq 1$, it has been shown that the codimension is minimal, in fact vanishes, iff the matrix pencil contains right null-space blocks only and the dimension of the null-space blocks are chosen equal to each other or, if this is not possible, the difference between the two different dimensions of nullspace blocks is at most $1$.

In Appendix \ref{app:generic} we show that the union of SLOCC classes of states corresponding to a generic set of matrix pencils is of full measure. More precisely, we use the lemmata and the theorem above to prove there the following theorem.

\begin{theorem}\label{theo:genstates}
The set of full rank states in $2\times m \times n$ belonging to a SLOCC class with a representative whose corresponding matrix pencil is generic, is of full measure.
\end{theorem}

Stated differently, we have that for $n=m$ a generic set of states is given by the union of SLOCC classes whose representatives are given by
\bea
\label{GSneqm}
\ket{\Psi(x_1,\ldots,x_m)}=\ket{0} (D_1\otimes \one)\ket{\Phi_m^+}+ \ket{1} (\one \otimes \one)\ket{\Phi_m^+},
\eea

where $D_1=\operatorname{diag}(x_1,\ldots,x_m)$, where $x_i \neq x_j$ for $i \neq j$. Note that not all states $\ket{\Psi(x_1,\ldots,x_m)}$ correspond to different SLOCC classes, as, e.g., the entries of $D_1$ could be sorted differently even via local unitaries. Moreover, as mentioned before, the eigenvalues of the matrix pencil, $x_i$ can be altered by an operator applied by $A$. More precisely, three of the eigenvalues can be fixed, leading to the fact that the representatives constitute a $m-3$ parameter family. For $n>m$ a generic set of states is given by the SLOCC class of a single state corresponding to the matrix pencil given in Theorem \ref{ThGenPencils}. For instance, in case $m=7, n=10$, the generic matrix pencil is $L_2\oplus L_2\oplus L_3$ and the representative of the generic SLOCC class is given by
\bea
\label{egnneqm}
\ket{\Psi}=\ket{0} (\ket{01}+\ket{12}+\ket{24}+\ket{35}+\ket{47}+\ket{58}+\ket{69})+ \ket{1} (\ket{00}+\ket{11}+\ket{23}+\ket{34}+\ket{46}+\ket{57}+\ket{68}).
\eea

For $n=m+1$, i.e. $d=1$ the generic matrix pencil is the single nullspace block, $L_m$. Hence, in this case the generic SLOCC class is represented by the state

\bea
\label{GSnneqm}
\ket{\Psi}=\ket{0} \left(\sum_{i=1}^{m} \ket{i-1,i} \right) + \ket{1} \ket{\Phi_m^+}.
\eea

Note that there is only one generic SLOCC class in case $n\neq m$. However, there is a $(m-3)$--parameter family of SLOCC--classes in case $m=n$ whose union constitutes a generic set of states.

\section{Transformations from an arbitrary state in a full measure set in $2\times m \times n$  to a full measure set in $2\times m \times (n-1)$}

We consider here state transformations via local, however not invertible matrices. That is, we consider transformations from e.g., $2\times m\times n$ to $2\times m\times (n-1)$. We show here that a generic state in $2\times m\times n$ can be transformed into any state in a full measure set within $2\times m\times (n-1)$ for any $m,n$. In order to do so, we first show that the result is true for $n=m+1$ (Lemma \ref{lemmanulltogeneric}). There, the statement is proven by showing that any state in the generic set of states within $2\times m\times m$ can be reached. Hence, we need to show that any state which corresponds to a SLOCC class whose representative is given by Eq. (\ref{GSneqm}), i.e., corresponds to a matrix pencil with $m$ distinct eigenvalues, can be reached from a state in $2\times m\times (m+1)$ corresponding to the SLOCC class represented by the state given in Eq. (\ref{GSnneqm}). The second exceptional case, which will be studied separately here, concerns the transformations from $2\times m\times m$ to 
$2\times m\times (m-1)$. As both cases involve transformation to or from generic states in $2\times m\times m$, we treat them separately before we present the general case where we consider transformations from $2\times m\times n$ to $2\times m\times (n-1)$ for $n$ arbitrary and show that any generic state in $2\times m\times n$ can be transformed into a full measure set of states in $2\times m\times (n-1)$ (see Theorem \ref{GenericToGeneric} and Figure \ref{fig:generichierarchy}). In fact, we prove an even stronger result, which shows that also some non--generic states can be reached.

Let us first study transformations from $2\times m\times( m+1)$ to $2\times m\times m$. Let us consider a state $\ket{\psi}$ with local ranks $2 \times m \times (m+1)$ corresponding to a pencil that is strictly equivalent to the matrix pencil containing only a single block $L_m$. We will prove that $\ket{\psi}$ can always be transformed to states with local ranks $2 \times m \times m$ that correspond to matrix pencils with $m$ distinct eigenvalues. Let us remark here that these are not the only states that can be reached from $\ket{\psi}$. In particular, it is also possible to reach certain states corresponding to matrix pencils with degenerate eigenvalues. We elaborate on that after proving the following Lemma.

\begin{lemma}\label{lemmanulltogeneric} A $2 \times m \times (m+1)$ state $\ket{\psi}$ corresponding to the matrix pencil $\mathcal{P}_\psi = L_m$ can be mapped to any $2 \times m \times m$ state $\ket{\phi}$ corresponding to a (diagonal) matrix pencil with arbitrary $m$ distinct eigenvalues $x_1, \ldots, x_m$, i.e. $\mathcal{P}_\phi = \bigoplus_{i=1}^{m} J(x_i)$, via local operators.
\end{lemma}
\begin{proof}
We consider the matrix pencil consisting of a single $L_m$ block and use Theorem \ref{theorem:traforules} to prove the statement. Due to Theorem \ref{theorem:traforules} there exist local (non-invertible) transformations which map $L_m$ to the matrix pencil where the last column of $L_m$ is added to the others with coefficients $-a_0,-a_1, \ldots, -a_{m-1} \in \mathbb{C}$ that we will choose later on. That is, there exist local operators accomplishing the transformation 

\begin{equation}\label{nulltogeneric}
L_{m}=\overset{m+1}{\left.\overleftrightarrow{\begin{pmatrix}\lambda & \mu & \cdot & \cdot & \cdot & \cdot\\
\cdot & \lambda & \mu & \cdot & \cdot & \cdot\\
\cdot & \cdot & \lambda & \mu & \cdot & \cdot\\
\cdot & \cdot & \cdot & \ddots & \ddots & \cdot\\
\cdot & \cdot & \cdot & \cdot & \lambda & \mu
\end{pmatrix}}\right\updownarrow m}\quad\longmapsto\quad \mathcal{P}_m=\overset{m}{\left.\overleftrightarrow{\begin{pmatrix}
\lambda & \mu & \cdot & \cdot & \cdot \\
\cdot & \ddots & \ddots & \cdot & \cdot\\
\cdot & \cdot & \lambda & \mu & \cdot\\
\cdot & \cdot & \cdot & \lambda & \mu\\
-a_0 \mu & \cdots & \cdot & -a_{m-2} \mu & -a_{m-1} \mu + \lambda &
\end{pmatrix}}\right\updownarrow m}.
\end{equation}

In Section \ref{sec:kcfcomputation} we showed that the KCF of the resulting matrix pencil, $P_m$, is given by Eq. (\ref{equ:companionkcf}). In particular, it has been shown there that $D_k=1$ for all $k \leq m-1$ and $D_m(\mu,\lambda)$ is given by Eq. (\ref{eq:dmwitha}). Using that $D_m(\mu=0, \lambda) \neq 0$ which implies that there exist only finite eigenvalues, $x_1, \ldots, x_m$, we have seen in Section \ref{sec:kcfcomputation} that the coefficients, $a_0, a_1,\dots, a_{m-1}$, are given by
 \begin{equation}
a_k = (-1)^{m-k} \bigg(\sum_{\underset{\sum_{j=1}^{m}i_{j}=(m-k)}{i_{1},i_{2},\dots,i_{m}\in\{0,1\}}}x^{i_1}_1x^{i_2}_2\dots x^{i_m}_m \bigg).
 \end{equation}
This tells us how to choose appropriate coefficients $a_0, \ldots, a_{m-1}$ to reach states that correspond to a pencil with finite eigenvalues, $x_1, x_2,\dots, x_m$ and with a KCF given in Eq. (\ref{equ:companionkcf}). In particular, we can reach any state corresponding to a pencil with distinct, finite eigenvalues. Due to Theorem \ref{theo:fractional} (see also discussion above Theorem \ref{theo:fractional}) infinite eigenvalues do not need to be considered, as the corresponding states are always SLOCC equivalent to a state with finite eigenvalues. This completes the proof.
\end{proof}

From the proof above we can see that it is also possible to reach states whose corresponding pencils have non-distinct eigenvalues. As $D_{m-1}=1$, however, the corresponding size signatures cannot be chosen freely, they are given by $s_i = (0, \ldots, 0, m_i)$, where $m_i$ equals the algebraic multiplicity of eigenvalue $x_i$.

Let us remark here, that the actual map $A \otimes B \otimes C$ between the two states $\ket{\psi}$ and $\ket{\phi}$ from Lemma \ref{lemmanulltogeneric} given in SKCF can be easily constructed as follows.
First note that the operator $C$ which is given in Eq. (\ref{eq:clairemap}), acting on the state $\ket{\psi}$, transforms the corresponding matrix pencil as in Eq. (\ref{nulltogeneric}). We now have to find local operators bringing the pencil $P_m$ to KCF. As explained in Section \ref{sec:kcfcomputation}, because of the special form of $P_m$ and as the eigenvalues, $x_i$, are distinct, we can use the Vandermonde matrix given in Eq. (\ref{EqVandermonde}) to diagonalize the matrix pencil $P_m$. The local operators bringing $P_m$ to KCF are then $\identity \otimes {V^{-1}}^{T} \otimes V$. In case the matrix pencil corresponding to the target state has infinite eigenvalues, additional operators $A \otimes B \otimes C$ have to be used which bring the matrix pencil to a form with infinite eigenvalues as explained in Section \ref{sec:alicedetails}. Concatenating all the operations yields the required operators.

Let us now, conversely, start from an arbitrary state, $\ket{\psi}$, with local ranks $2 \times (m+1) \times (m+1)$ corresponding to a matrix pencil with $m+1$ distinct eigenvalues (see Eq. (\ref{GSneqm})) and consider transformations to states with local ranks $2 \times m \times (m+1)$. Here, we prove that a state $\ket{\phi}$ with $\mathcal{P}_\phi = L_{m}$ can be reached from any such state $\ket{\psi}$. Note that in contrast to before, we now consider a scenario where Bob instead of Claire applies a non-invertible operation. We do so in order to obtain the stair hierarchy presented in Figure \ref{fig:stairs}, but, obviously, one can also show that a state corresponding to the pencil $L^T_{m}$ can be reached from $\ket{\psi}$ if a non-invertible operation on Claire's side is considered instead.

\begin{lemma}\label{mainpencil} Any state $\ket{\psi}$ in $2\times( m+1)\times (m+1)$ corresponding to a $(m+1) \times (m+1)$ matrix pencil $\mathcal{P}_\psi$ with $m+1$ distinct eigenvalues can be transformed via local operations to the $2\times m\times (m+1)$ state $\ket{\phi}$ corresponding to a matrix pencil that consists of one single $L_m$ block, i.e., $\mathcal{P}_\phi = L_m$.
\end{lemma}
\begin{proof}
Let us denote the eigenvalues of $\mathcal{P}_\psi$ as $x_1, x_2, \ldots, x_m$. Due to Theorem \ref{theorem:traforules} there exists a non-invertible matrix $B$ which, applied to $\ket{\psi}$, transforms $\ket{\psi}$ to a state whose corresponding matrix pencil, $\mathcal{P'}$, can be obtained by erasing the first row of the initial pencil after adding it to the other rows with coefficients $a_2, a_3, \ldots, a_{m+1}$. We will prove that the resulting matrix pencil, $\mathcal{P'}$, is strictly equivalent to $\mathcal{P}_\phi = L_m$ for some choice of the coefficients $a_2, a_3, \ldots, a_{m+1}$. The explicit pencils read
\begin{equation}\label{mainpencilproof}
\mathcal{\mathcal{P}_\psi}={\begin{pmatrix}\mu x_{1}+\lambda & 0 & \cdot & \cdot\\
\cdot & \mu x_{2}+\lambda & \ddots & \cdot\\
\cdot & \cdot & \ddots & 0\\
\cdot & \cdot & \cdot & \mu x_{m+1}+\lambda
\end{pmatrix}} \text{ and } \mathcal{\mathcal{P}}'=\overset{m+1}{\left.\overleftrightarrow{
\begin{pmatrix}
a_{2}(\mu x_{1}+\lambda) & \mu x_{2}+\lambda & \cdot & \cdot & \cdot \\
a_{3}(\mu x_{1}+\lambda) & \cdot & \mu x_{3}+\lambda & \cdot & \cdot\\
\vdots & \cdot & \cdot & \ddots & \cdot\\
a_{m+1}(\mu x_{1}+\lambda) & \cdot & \cdot & \cdot & \mu x_{m+1}+\lambda
\end{pmatrix}}\right\updownarrow m}.
\end{equation}
It is straightforward to see that $D_m$ of $\mathcal{P}'$ is given by 
\begin{align}
D_m=\quad \operatorname{gcd} & \{(\mu x_{2}+\lambda)(\mu x_{3}+\lambda)\cdots (\mu x_{m+1}+\lambda),\nonumber\\
-&a_2(\mu x_{1}+\lambda)(\mu x_{3}+\lambda)\cdots (\mu x_{m+1}+\lambda),\nonumber\\
 &a_3(\mu x_{1}+\lambda)(\mu x_{2}+\lambda)(\mu x_{4}+\lambda)\cdots (\mu x_{m+1}+\lambda), \nonumber\\
&\vdots\nonumber\\
 (-1)^{m} &a_{m+1}(\mu x_{1}+\lambda)(\mu x_{2}+\lambda)\cdots (\mu x_{m}+\lambda)\}\nonumber\\
=\quad \operatorname{gcd}&\{(-1)^{j+1} a_j \prod_{i=1, i\neq j} ^{m+1} (\mu x_i +\lambda)\quad | \quad j=1, \dots , (m+1)\}=1,
\end{align}
where $a_1 = 1$. Note that the last equality holds iff none of the coefficients $a_i$ vanish. This is in fact the only requirement we impose on the choice of coefficients. Using Lemma \ref{lemma:dmequalone} it follows that $\mathcal{P}'$ is strictly equivalent to $L_m$, which proves the statement.
\end{proof}

\begin{figure}[h!]
\includegraphics[scale=0.6]{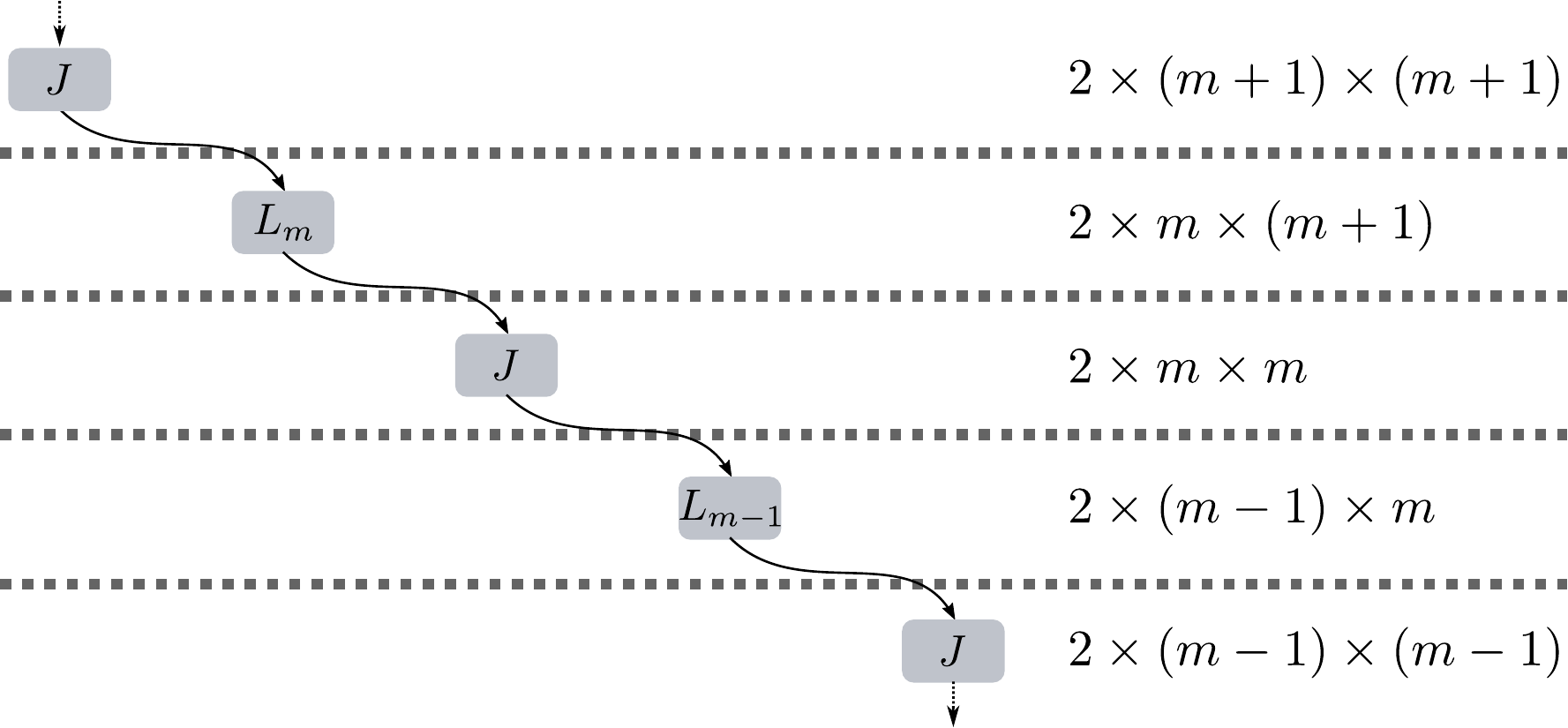}
\caption{Schematic representation of some SLOCC convertibility relations among states in the neighbouring dimensions $2 \times (m + 1) \times (m + 1)$, $2 \times m \times (m+1)$, $2 \times m \times m$, $2 \times (m-1) \times m$, and $2 \times (m-1) \times (m-1)$. The transformation rules proven in Lemma \ref{lemmanulltogeneric} and Lemma \ref{mainpencil} are concatenated to obtain this ``stair'' hierarchy.}
\label{fig:stairs}
\end{figure}

Using the lemmata above, we are now in the position to prove one of the main results of this paper, namely that a generic state in $2\times m \times  n$ can be transformed to any state in the generic set of states in $2\times m \times ( n-1)$, as stated in the following theorem.

\begin{theorem} 
\label{GenericToGeneric}
Any generic state $\ket{\psi}$ in $2 \times m \times n$ can be transformed locally to any generic state $\ket{\phi}$ in $2 \times m \times (n-1)$ for any $m$ and $n$.
\end{theorem}

In order to prove this theorem, we make use of the following lemma, which is actually more general than needed for the proof of the theorem above.

\begin{lemma}
\label{lemma:Ldistribution}
A state $\ket{\psi}$ in $2 \times m \times n$ with $\mathcal{P}_\psi = \bigoplus_{i=1}^{n-m} L_{\epsilon_i}$ can be transformed to a state $\ket{\phi}$ in $2 \times m \times (n-1)$ with $\mathcal{P}_\phi = \bigoplus_{i=1}^{n-m-1} L_{\epsilon'_i}$ via local operations for any $n \geq m + 2$ if the following condition holds.
There exists $j \in \{1, \ldots, n-m-1\}$ such that for all $i \in \{1, \ldots, j-1\}$ $\epsilon_i = \epsilon_i'$ and for all $i \in \{j, \ldots, n-m-1\}$ $\epsilon_{i+1} \leq \epsilon_i'$, where we assume $(\epsilon_i)_i$ and $(\epsilon'_i)_i$ to be sorted in ascending order.
\end{lemma}

The proof of this lemma can be found in Appendix \ref{app:Ldistribution}. Let us now use it to prove Theorem \ref{GenericToGeneric}.

\begin{proof}[Proof of Theorem \ref{GenericToGeneric}]
Let us write $d = n - m$, where $d\geq 0$. The cases $d=0$ and $d=1$ were discussed and proven above in Lemma \ref{mainpencil} and Lemma \ref{lemmanulltogeneric}, respectively. Let us now prove the statement for the case $d\geq 2$. To this end, we consider the matrix pencils $\mathcal{P}_{\psi}$ and $\mathcal{P}_{\phi}$ corresponding to the states $\ket{\psi}$ and $\ket{\phi}$ which we assume w.l.o.g. to be in SKCF.
As shown in Theorem \ref{theo:genstates}, the matrix pencils corresponding to generic states are given by 
$\mathcal{P}_\psi = (d-(m\ \operatorname{mod}\ d))L_{\left\lfloor m/d \right\rfloor}\oplus (m\ \operatorname{mod}\ d)L_{\left\lceil m/d\right\rceil }$
and $\mathcal{P}_\phi = (d-1-(m\ \operatorname{mod}\ (d-1))L_{\left\lfloor m/(d-1) \right\rfloor}\oplus (m\ \operatorname{mod}\ (d-1))L_{\left\lceil m/(d-1)\right\rceil }$.
Note that the matrix pencils contain only $L_{\epsilon_i}$ blocks. More preciscely,
 $\epsilon_1 = \ldots = \epsilon_{(d-(m\ \operatorname{mod}\ d))} = \left\lfloor m/d \right\rfloor$, and if $m\ \operatorname{mod}\ d \neq 0$, $\epsilon_{(d-(m\ \operatorname{mod}\ d) + 1 )} = \ldots = \epsilon_d = \left\lceil m/d\right\rceil$ for $\mathcal{P}_\psi$ and  
   $\epsilon'_1 = \ldots = \epsilon'_{(d-1-(m\ \operatorname{mod}\ (d-1)))} = \left\lfloor m/(d-1) \right\rfloor$, and if $m\ \operatorname{mod}\ (d-1) \neq 0$, $\epsilon'_{(d-1-(m\ \operatorname{mod}\ (d-1)) + 1 )} = \ldots = \epsilon'_{(d-1)} = \left\lceil m/(d-1)\right\rceil$ for $\mathcal{P}_\phi$. Hence, the sizes of the blocks are distributed in such a way that Lemma \ref{lemma:Ldistribution} applies proving that $\ket{\psi}$ can be transformed to $\ket{\phi}$, which completes the proof.

\end{proof}
\subsection*{Examples of transformations from generic states to lower-dimensional generic states   }

Let us now consider a simple example of such a transformation. Consider the two generic $2 \times 7 \times 10$ and $2 \times 7 \times 9$ states  
\begin{align}
   \ket{\psi}=\ket{0} (\ket{01}+\ket{12}+\ket{24}+\ket{35}+\ket{47}+\ket{58}+\ket{69})+ \ket{1} (\ket{00}+\ket{11}+\ket{23}+\ket{34}+\ket{46}+\ket{57}+\ket{68}) \text{ and}
\end{align}
\begin{align}
   \ket{\phi}=\ket{0} (\ket{01}+\ket{12}+\ket{23}+\ket{35}+\ket{46}+\ket{57}+\ket{68})+ \ket{1} (\ket{00}+\ket{11}+\ket{22}+\ket{34}+\ket{45}+\ket{56}+\ket{67}),
\end{align}
respectively. The matrix pencils corresponding to these states are $\mathcal{P}_\psi = L_2 \oplus L_2\oplus L_3$ and $\mathcal{P}_\phi= L_3\oplus L_4$. We will now show how the operators $B$ and $C$ such that $B \mathcal{P}_\psi C^T = \mathcal{P}_\phi$ and thus $\identity \otimes B \otimes C \ket{\psi} = \ket{\phi}$ are explicitly constructed in this example. In order to do so, we use the proof of Lemma \ref{lemma:Ldistribution}, which is presented in Appendix \ref{app:Ldistribution}. The matrices $C^T$ and $\tilde{B}$ as defined there (see Eq. (\ref{eq:mapc}) and Eq. (\ref{eq:mapb})) are
\begin{align}
\label{eq:gentrafoexample}
C^T=\begin{pmatrix}
1 & \cdot & \cdot & \cdot & \cdot & \cdot & \cdot & \cdot & \cdot\\
\cdot & 1 & \cdot & \cdot & \cdot & \cdot & \cdot & \cdot & \cdot\\
\cdot & \cdot & 1 & \cdot & \cdot & \cdot & \cdot & \cdot & \cdot\\
\cdot & 1 & \cdot & \cdot & 1 & \cdot & \cdot & \cdot & \cdot\\
\cdot & \cdot & 1 & \cdot & \cdot & 1 & \cdot & \cdot & \cdot\\
\cdot & \cdot & \cdot & 1 & \cdot & \cdot & 1 & \cdot & \cdot\\
\cdot & \cdot & \cdot & \cdot & \cdot & 1 & \cdot & \cdot & \cdot\\
\cdot & \cdot & \cdot & \cdot & \cdot & \cdot & 1 & \cdot & \cdot\\
\cdot & \cdot & \cdot & \cdot & \cdot & \cdot & \cdot & 1 & \cdot\\
\cdot & \cdot & \cdot & \cdot & \cdot & \cdot & \cdot & \cdot & 1\\
\end{pmatrix}, \
\tilde{B}=\begin{pmatrix}
1 & \cdot & \cdot & \cdot & \cdot & \cdot & \cdot\\
\cdot & 1 & \cdot & \cdot & \cdot & \cdot & \cdot\\
\cdot & 1 & \cdot & 1 & \cdot & \cdot & \cdot\\
\cdot & \cdot & 1 & \cdot & 1 & \cdot & \cdot\\
\cdot & \cdot & \cdot & \cdot & 1 & \cdot & \cdot\\
\cdot & \cdot & \cdot & \cdot & \cdot & 1 & \cdot\\
\cdot & \cdot & \cdot & \cdot & \cdot & \cdot & 1\\
\end{pmatrix}.
\end{align}
It can be easily verified that $\tilde{B}$ is invertible and that indeed $B \mathcal{P}_\psi C^T = \mathcal{P}_\phi$, where $B = \tilde{B}^{-1}$.

Let us also present a different method to show that the matrix pencil $\mathcal{P}_\psi C^T$ is indeed equivalent to $\mathcal{P}_\phi$.
To this end, let us consider the action of $C^T$ on the nullspace vectors of the matrix pencil $\mathcal{P}_\phi$.
Let $\{x_1, x_2, x_3\}$ ($\{x'_1, x'_2\}$) denote the nullspace vectors of minimal degree for $\mathcal{P}_\psi$ ($\mathcal{P}_\phi$), respectively. We have
\begin{equation}
x_{1}=\begin{pmatrix}\mu^{2}\\
-\mu\lambda\\
\lambda^{2}\\
0\\
0\\
0\\
0\\
0\\
0\\
0
\end{pmatrix},\ x_{2}=\begin{pmatrix}0\\
0\\
0\\
\mu^{2}\\
-\mu\lambda\\
\lambda^{2}\\
0\\
0\\
0\\
0
\end{pmatrix},\ x_{3}=\begin{pmatrix}0\\
0\\
0\\
0\\
0\\
0\\
\mu^{3}\\
-\mu^{2}\lambda\\
\mu\lambda^{2}\\
-\lambda^{3}
\end{pmatrix}\ \ \ \mbox{\ensuremath{\mbox{ and}}  }\ \ \ x'_{1}=\begin{pmatrix}\mu^{3}\\
-\mu^{2}\lambda\\
\mu\lambda^{2}\\
-\lambda^{3}\\
0\\
0\\
0\\
0\\
0
\end{pmatrix},\ x'_{2}=\begin{pmatrix}0\\
0\\
0\\
0\\
\mu^{4}\\
-\mu^{3}\lambda\\
\mu^{2}\lambda^{2}\\
-\mu\lambda^{3}\\
\lambda^{4}
\end{pmatrix}.
\end{equation}
The matrix $C^T$ is constructed in such a way, that the nullspace vectors $\{x_i'\}$ transform under it as follows
\begin{align}
C^T x_{1}'=\mu\begin{pmatrix}
\mu^{2}\\
-\mu\lambda\\
\lambda^{2}\\
0\\
0\\
0\\
0\\
0\\
0\\
0
\end{pmatrix}-\lambda \begin{pmatrix}0\\
0\\
0\\
\mu^{2}\\
-\mu\lambda\\
\lambda^{2}\\
0\\
0\\
0\\
0\\
\end{pmatrix} = \mu x_1 - \lambda x_2 \quad   \mbox{ and } \quad  C^T x_{2}'=\mu^2\begin{pmatrix}0\\
0\\
0\\
\mu^{2}\\
-\mu\lambda\\
\lambda^{2}\\
0\\
0\\
0\\
0\\
\end{pmatrix}-\lambda \begin{pmatrix}
0\\
0\\
0\\
0\\
0\\
0\\
\mu^{3}\\
-\mu^{2}\lambda\\
\mu\lambda^{2}\\
-\lambda^{3}
\end{pmatrix} = \mu^2 x_2 - \lambda x_3.
\end{align}
Note that $C^T x'_1$ ($C^T x'_2$) can be expressed as a superposition of the nullspace vectors $x_1$ and $x_2$ ($x_2$ and $x_3$), respectively. Note further that $C^T x'_1$ and $C^T x'_2$ are linearly independent. As $\mathcal{P}_\psi x_i = 0$, where $i \in \{1,2,3\}$, we have that $\mathcal{P}_\phi C^T x'_i = 0$, hence $x'_1$ and $x'_2$ are two linearly independent vectors in the nullspace of $\mathcal{P}_\psi C^T$ with degree 3 and 4, respectively. We will now show that these are the minimal indices of $\mathcal{P}_\psi C^T$, which implies that $\mathcal{P}_\psi C^T$ is strictly equivalent to $\mathcal{P}_\phi$. First note that $D_m=D_7=1$ for $\mathcal{P}_\psi C^T$, as there is one 7-minor equal to $\mu^7$ and another 7-minor equal to $\lambda^7$. Using Lemma \ref{lemma:dmequalone} this implies that the KCF of $\mathcal{P}_\psi C^T$ is a direct sum of right null space blocks only. Due to dimensionality reasons, the KCF contains exactly two $L$ blocks. In fact, the possibilities are $L_1 \oplus L_6$, $L_
2 \oplus L_5$, and $L_3 \oplus L_4$. As the minimal indices are $(1,6)$, $(2,5)$, and $(3,4)$, respectively, it is easy to see that $\mathcal{P}_\psi C^T$ has to be strictly equivalent to $\mathcal{P}_\phi = L_3 \oplus L_4$, as the minimal indices corresponding to the other choices are not compatible with the degrees of $x'_1$ and $x'_2$ (see Lemma \ref{lemma:minimalindicesbound} in Appendix \ref{app:Ldistribution}).

Note that also in general, the matrix $C^T$ is constructed in such a way that the linearly independent vectors $C^T x'_i$ can be individually expressed as a superposition of $x_{i}$ and $x_{i+1}$ and thus $\mathcal{P}_\psi C^T x'_i =0$ for all $i \leq d-1$. Recall the proof of Lemma \ref{lemma:Ldistribution}, where we have shown that for a matrix pencil strictly equivalent to $\mathcal{P}_\psi C^T$, there exists a non-vanishing $m$-minor equal to $\mu^m$. One can transform $\mathcal{P}_\psi C^T$ to a upper block triangular matrix, which can be achieved in a iterative process similarly as described in the proof of Lemma \ref{lemma:Ldistribution}, where one starts from the last block, $L_{\epsilon_d}$, instead and obtains an $m$-minor equal to $\lambda^m$. $D_m$ is invariant under such a transformation and must thus divide both $\mu^m$ and $\lambda^m$. Hence, we have that $D_m = 1$. As in the example above, Lemmata \ref{lemma:dmequalone} and \ref{lemma:minimalindicesbound} can be used to show that $\mathcal{P}_\psi C^T$ is strictly equivalent to $\mathcal{P}_\phi$ and, thus, obtain an alternative proof of Lemma \ref{lemma:Ldistribution}. 

Using the matrices $\tilde{B}$ and $C^T$ given in Eq. (\ref{eq:gentrafoexample}), one finds the following operations $B = \tilde{B}^{-1}$ and $C$ such that $\identity \otimes B \otimes C \ket{\psi} = \ket{\phi}$, where
\begin{align}
   C = \ket{0}\bra{0} + \ket{1}\bra{1} + \ket{2}\bra{2} + (\ket{1}+\ket{4})\bra{3} + (\ket{2}+\ket{5})\bra{4} + (\ket{3}+\ket{6})\bra{5} + \ket{5}\bra{6} + \ket{6}\bra{7} + \ket{7}\bra{8} + \ket{8}\bra{9},
   \end{align}  
\begin{align}
   B = \ket{0}\bra{0} + (\ket{1}-\ket{3})\bra{1} + \ket{3}\bra{2} + \ket{2}\bra{3} + (-\ket{2}+\ket{4})\bra{4} + \ket{5}\bra{5} + \ket{6}\bra{6}. 
\end{align} 

With Theorem \ref{GenericToGeneric} we have thus proven that the SLOCC hierarchy for generic states in $2 \times m \times n$ is of the form shown in Figure \ref{fig:generichierarchy}. In particular, any generic state can be converted to a generic set of states contained in an lower-dimensional Hilbert space.

\begin{figure}[h!]
\includegraphics[scale=0.7]{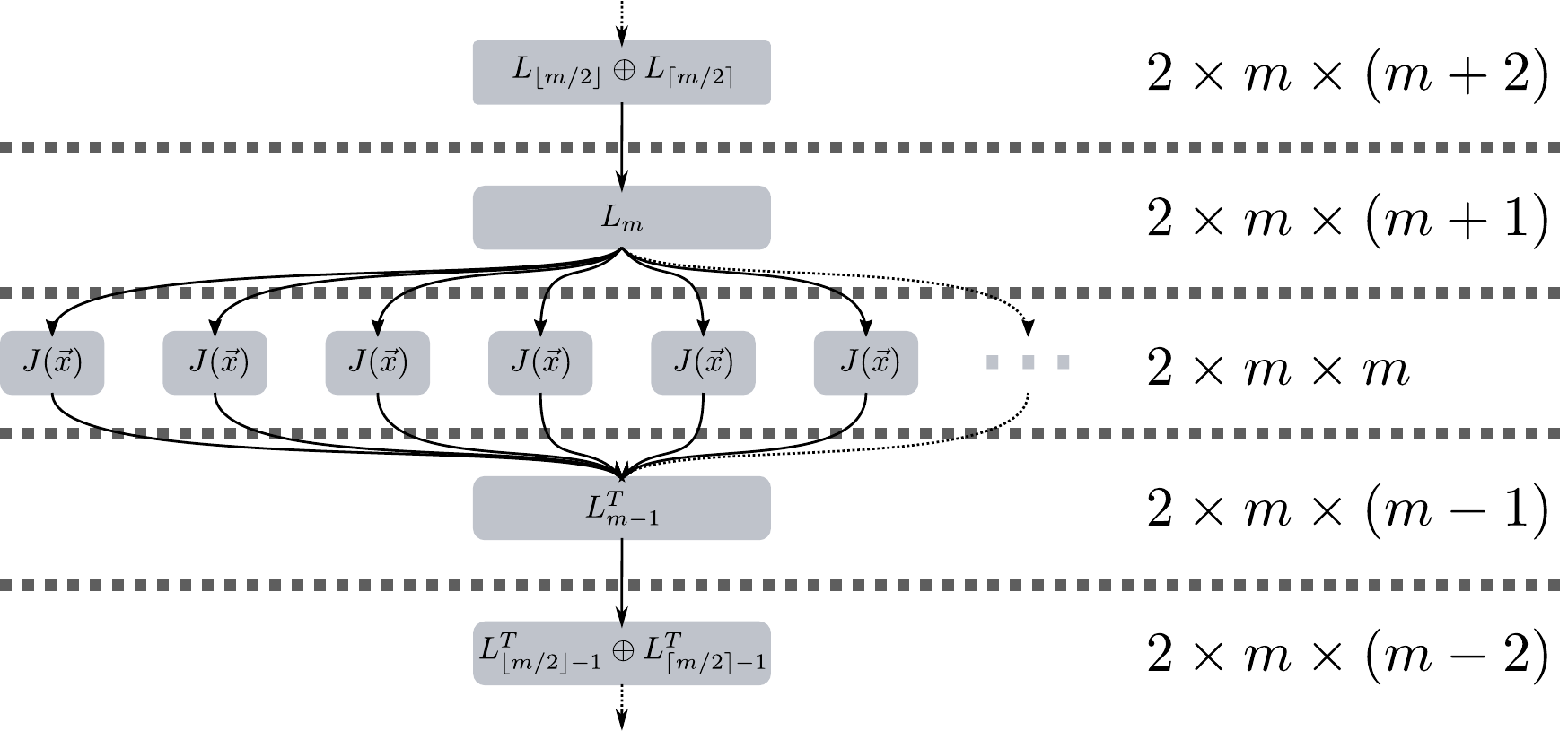}
\caption{Schematic representation of the hierarchy of SLOCC classes for generic states in $2 \times m \times n$ level systems for fixed $m$. The here depicted matrix pencils correspond to representative states of the generic SLOCC classes in case $m \neq n$. In case $m=n$, the union of the SLOCC classes corresponding to the depicted matrix pencils forms the generic set of states.
 A similar picture arises if $n$ is fixed instead.}
\label{fig:generichierarchy}
\end{figure}

\section{Common resource states}
\label{sec:cr}
In the previous section we showed that any generic state in $2\times m\times (n+1)$ can be transformed probabilistically to any generic state in $2\times m\times n$. Stated differently, it is sufficient to increase $n$ only by one in order to find a generic state that can be transformed to every  generic state in a lower dimensional Hilbert space. We call such a state a \textit{resource state} for a set of states in $2\times m\times n$, which it can be transformed to. Note that here a state is a representative of its SLOCC class in its respective dimension. A formal definition of resource states and their optimality is given below. 
The notion of optimal common resource states has been introduced and motivated in \cite{GuCh16}. The general idea is that for some quantum information task, different entangled states might be needed. Hence, if the parties would share the common resource of the required states, they could obtain any of the required states when needed by LOCC operations only. In contrast to \cite{GuCh16}, where resource states under deterministic LOCC are considered, in this work, we consider resource states under probabilistic transformations.

\begin{definition} 
A common resource (CR) state for a set of states $S$ is a state that can be transformed probabilistically to any state contained in $S$. A common resource state $\ket{\psi}$ is called \emph{optimal} common resource state (OCR), if for any other common resource state $\ket{\phi}$ for $S$, it holds that either $\ket{\phi}$ can be probabilistically transformed to $\ket{\psi}$, or neither $\ket{\psi}$ nor $\ket{\phi}$ can be probabilistically transformed into the other.
\end{definition}
Note that a common resource state can only be found in a higher-dimensional Hilbert space unless we consider the trivial case where some state in $S$ already serves as common resource for $S$. 
Note further that the optimal common resource state is not necessarily unique. This is due to the fact that either one might increase the dimensions in different ways to find common resource states, or, even when considering a fixed increase of dimensions, one might find common resource states in different SLOCC classes.
To give an example, consider $2 \times 2 \times 3$ states (see \cite{ChMi10}). A optimal common resource for $\mathcal{H} = \mathbb{C}^2\times \mathbb{C}^2\times\mathbb{C}^3$ can be found in dimensions $2 \times 2 \times 4$ as well as in $2 \times 3 \times 3$. Moreover, in $2 \times 3 \times 3$ one can in fact find three SLOCC-inequivalent common resources, which are all optimal. However, any common resource state found in $2 \times 3 \times 4$ 
can be transformed to a common resource of lower dimension and hence cannot be optimal.

It has been shown that such a common resource state always exists if one of the dimensions is increased sufficiently. In particular in \cite{DuSh09}, it has been shown that given a Hilbert space $\mathcal{H} = \mathbb{C}^{d_1}\otimes \mathbb{C}^{d_2}\otimes \dots \otimes \mathbb{C}^{d_N}$, where $d_1\leq d_2\leq\dots \leq d_N$, there exists a state which can be transformed to all the other states by LOCC iff 
\begin{align}\label{DuanEq}
\prod_{i=1}^{N-1}d_i\leq d_N.
\end{align}
In case Eq. (\theequation) holds, the $N$-th party (with the highest local dimension) can share a $d_i$-dimensional singlet state with each $i$-th party. In order to obtain any state, the $N$-th party can then locally prepare the desired state and perform a teleportation protocol using the shared singlet states. This upper bounds the increase in dimensions that has to be considered when looking for common resource states. 

In this section, we extend our result from the previous section, where we considered generic sets of states. In particular, we look for common resource states for $S$ containing all states in $\mathcal{H} = \mathbb{C}^2\times \mathbb{C}^m\times\mathbb{C}^ n$ for fixed $m$ and $n$, i.e., where in addition to the generic set of states also sets of measure zero are included.
In the following, we show that in contrast to common resources for a generic set of states, where in order to find a resource state it suffices to increase some dimension by one, this increase is not sufficient to find common resource states for the full set of states. 
To find such common resource states, we consider an increase of the highest dimension, $n$, i.e., we look for common resource states in $\mathbb{C}^2\times \mathbb{C}^m \times \mathbb{C}^{\tilde{n}}$, where $\tilde{n} > n$ \footnote{Note that one could as well consider an increase in $m$, or an increase in both $m$ and $n$. This is however beyond the scope of this paper.}. 
Due to the above mentionend bound, it is clear that a common resource state exists in case $\tilde{n} = 2m$, as then the condition in Eq. (\theequation) is satisfied. However, the question we address here is whether we can find a common resource state for which $\tilde{n} < 2m$. Again, a difference between the two cases $m=n$ and $m \neq n$ becomes apparent. In the following theorems, we prove that in case $m=n$ such a common resource state exsits for $\tilde{n} = 2m-2$. Moreover, this increase of dimension (of the third system) is optimal. In contrast to that, in case $m \neq n$, the teleportation bound is tight, i.e., there exists no common resource state whenever $\tilde{n} < 2m$.

The following $m \times (m-2)$ matrix pencil will be of relevance throughout the remainder of this section,
\begin{align}
\label{eq:mmocr}
\mathcal{P}_\psi= (m-3) L_1 \oplus L_2 \oplus M^1(0) = \begin{pmatrix}
\lambda & \mu & \cdot & \cdot & \cdot & \cdot & \cdot & \cdot & \cdot & \cdot & \cdot & \cdot \\
\cdot & \cdot & \lambda & \mu & \cdot & \cdot & \cdot & \cdot & \cdot & \cdot & \cdot & \cdot \\
\cdot & \cdot & \cdot & \cdot & \ddots & \ddots & \cdot & \cdot & \cdot & \cdot & \cdot & \cdot \\
\cdot & \cdot & \cdot & \cdot & \cdot & \cdot & \lambda & \mu & \cdot & \cdot & \cdot & \cdot \\
\cdot & \cdot & \cdot & \cdot & \cdot & \cdot & \cdot & \cdot & \lambda & \mu & \cdot & \cdot \\
\cdot & \cdot & \cdot & \cdot & \cdot & \cdot & \cdot & \cdot & \cdot & \lambda & \mu & \cdot \\
\cdot & \cdot & \cdot & \cdot & \cdot & \cdot & \cdot & \cdot & \cdot & \cdot & \cdot &  \lambda\\
\end{pmatrix}.
\end{align}
The reason for that is that this pencil corresponds to the $2 \times m \times (2m-2)$ state $\ket{\psi}$ which turns out to be a common resource for all $2 \times m \times m$ states, which we will show in the following theorem. 

\begin{theorem}
\label{theo:mmcr}
The $2 \times m \times (2m-2)$ state $\ket{\psi}$ that corresponds to the matrix pencil given in Eq. (\ref{eq:mmocr}) is a common resource for $2\times m\times m$ states, where $m \geq 4$, i.e., $\ket{\psi}$ can be transformed to any $2\times m\times m$ state via non-invertible SLOCC.
\end{theorem}

\begin{proof}
Although we could directly prove the statement considering the matrix pencil given in Eq. (\ref{eq:mmocr}), we will prove this theorem in three steps for readability. First, we will explicitly give a $2\times m\times 2m$ state which reaches all $2\times m\times m$ states. Based on this state, we will then construct a $2\times m\times (2m-1)$ state which can be used for the same task. Finally, we will argue that the dimension can be reduced by one more by showing that the $2\times m\times (2m-2)$ corresponding to the matrix pencil $\mathcal{P}_\psi$ given in Eq. (\ref{eq:mmocr}) is a common resource for all $2\times m\times m$ states.

A common resource $\ket{\psi}$ for $2\times m\times m$ states can be trivially found in $2\times m\times 2m$, namely $\ket{\phi^+_2}_{AC} \otimes \ket{\phi^+_m}_{BC}$. Associating $\ket{0}_C \otimes \ket{i}_C$ to $\ket{i}_C$ and $\ket{1}_C \otimes \ket{i}_C$ to $\ket{m+i}_C$, the corresponding matrix pencil is
\begin{align}\label{Eq:LOCCPencil}
 \mathcal{P}_\psi=\begin{pmatrix}
\mu & \cdot & \cdot & \lambda & \cdot & \cdot\\
\cdot & \ddots & \cdot & \cdot & \ddots & \cdot\\
\cdot & \cdot & \mu & \cdot & \cdot & \lambda\\
\end{pmatrix},
\end{align}
which is strictly equivalent to its KCF $m L_1$. As explained in the introduction of this section, it can reach all states in $2 \times m \times m$ (even by deterministic LOCC), as Claire can use the state to perform teleportation. 

We will present now an alternative proof of this statement by considering the corresponding matrix pencils.
Let us assume that $\ket{\psi}$ is in SKCF. We will show that given the matrix pencil $\mathcal{P}_\psi = m L_1$, any $m \times m$ matrix pencil $\mathcal{P}_\phi$ can be obtained applying Theorem \ref{theorem:traforules} $m$ times consecutively. In particular, we will explicitly show that the $L_1$ blocks in $\mathcal{P}_\psi$ can be combined in such a way that arbitrary blocks, $L$, $L^T$, and $J$ can be created for $\mathcal{P}_\phi$. Furthermore, we will count the number of required $L_1$ blocks and see that $m$ of them present in $\mathcal{P}_\psi$ suffice.

First note that $L_2$ can be obtained by adding the first column of some $L_1$ to the last column of another $L_1$ block and discarding this column afterwards. This procedure can be iterated to generate $L_\epsilon$, consuming $\epsilon$ $L_1$ blocks.
An $L^T_\nu$ block can be generated by first generating $L_{\nu+1}$ and then discarding the first and the last column of $L_{\nu+1}$. This procedure consumes $\nu+1$ of the $L_1$ blocks in total.
As the matrix pencil $\mathcal{P}_\phi$ is square and we do only consider transformations to pencils for which it holds that $h=0$ and $g=0$ in their KCF, we have that the number of $L$ blocks equals the number of $L^T$ blocks in $\mathcal{P}_\phi$. Let us assume w.l.o.g. that the matrix pencil $\mathcal{P}_\phi$ consists of null-space blocks of total dimension $k \times k$, where $k = \sum_{i=1}^{a} \epsilon_i + \sum_{i=1}^{a} (\nu_i + 1)$ and of a $J$ block of size $(m-k) \times (m-k)$.
As explained above, the null-space blocks can be created by consuming $k$ of the $L_1$ blocks of $\mathcal{P}_\psi$. We will now show that the $m-k$ remaining $L_1$ blocks can be used to create an arbitrary $J$ block. First note, that due to dimensionality reasons, the size signatures of the eigenvalues sum up to $m-k$. An arbitrary finite eigenvalue $x_i$ (infinite eigenvalue) with size signature 1 can be created consuming a single $L_1$ block by adding the second column to the first one with coefficient $x_i$ and discarding the second column afterwards (discarding the first column), respectively. Moreover, an existing $M^{e^i_{j}}(x_i)$ block can be enlarged to $M^{e^i_{j}+1}(x_i)$ consuming a single $L_1$ block arranged to the upper left, i.e., $L_1 \oplus M^{e^i_{j}}(x_i)$. To this end, the second column of $L_1$ is added to the first column of the existing $M^{e^i_{j}}(x_i)$ block before it is added to the first column of the $L_1$ block with coefficient $x_i$ and discarded afterwards. Similarly, an $N^{e^\mu_{j}}$ block can be enlarged. Thus, in total $m-k$ $L_1$ blocks are consumed in order to create an arbitrary $J$ block of size $(m-k) \times (m-k)$. Hence, $\mathcal{P}_\psi = m L_1$ can be transformed to an arbitrary $m \times m$ matrix pencil and thus, the $2 \times m \times 2m$ state $\ket{\psi}$ can be transformed to any $2\times m \times m$ state.

Let us now show that the $2 \times m \times (2m-1)$ state $\ket{\psi}$, whose corresponding matrix pencil is
\begin{align}
\mathcal{P}_\psi= (m-1) L_1\oplus M^1(0) = \begin{pmatrix}
\lambda & \mu & \cdot & \cdot & \cdot & \cdot & \cdot\\
\cdot & \cdot & \ddots & \ddots & \cdot & \cdot & \cdot\\
\cdot & \cdot & \cdot & \cdot & \lambda & \mu & \cdot\\
\cdot & \cdot & \cdot & \cdot & \cdot & \cdot & \lambda\\
\end{pmatrix},
\end{align}
can perform the same task. Note that the first $m-1$ $L_1$ blocks are the same as before. However, the last $L_1$ block is replaced by $M^1(0)$. We will distinguish the two cases that $\mathcal{P}_\phi$ has eigenvalues and that $\mathcal{P}_\phi$ has no eigenvalues. In both cases we show that $\mathcal{P}_\psi$ given in Eq. (\theequation) can be transformed into $\mathcal{P}_\phi$.
In the first case, as a first step, Alice's action is used to transform the eigenvalue 0 present in $\mathcal{P}_\psi$ to some eigenvalue $x$ that is present in $\mathcal{P}_\phi$. From this point on, the same procedure as described above is used to increase the size signature of the eigenvalue $x$ or to obtain the remaining blocks present in $\mathcal{P}_\phi$. In the second case, the pencil $\mathcal{P}_\phi$ does not contain a $J$ block, hence it is a direct sum of $L$ and $L^T$ blocks only. Due to dimensionality reasons, at least one $L^T$ block is present in $\mathcal{P}_\phi$. It is easy to see that in this case, the $M^1(0)$ block in $\mathcal{P}_\psi$, can be used together with $\nu$ $L_1$ blocks to obtain $L_\nu^T$ by first using the $\nu$ $L_1$ blocks to create $L_{\nu}$ as described above and finally adding the $M^1(0)$ block below the last column of the created $L_{\nu}$ bock and discarding the first column of this block. From this point on, the procedure described in the $2 \times m \times 2m$ 
scenario can be used in order to obtain the remaining blocks present in $\mathcal{P}_\phi$.
 
Let us now show that we can further reduce the dimension by one by considering the $2 \times m \times (2m-2)$ state $\ket{\psi}$, whose corresponding matrix pencil $\mathcal{P}_\psi$ is given in Eq. (\ref{eq:mmocr}).
As it is more involved to see that $\ket{\psi}$ is indeed a common resource to all $2 \times m \times m$ states and, thus, the theorem holds, we will consider the following three classes of matrix pencils separately. Note that each matrix pencil $\mathcal{P}_\phi$ belongs to at least one of these classes,
\begin{enumerate}[(i)]
	\item  there is at least one $L_\epsilon$ block with $\epsilon \geq 2$ present in $\mathcal{P}_\phi$,
	\item  there is at least one $L_1$ block present in $\mathcal{P}_\phi$,
	\item  there is no $L$ block present in $\mathcal{P}_\phi$.
\end{enumerate}
Let us first deal with case (i). In that case, the block $L_2$ together with $\epsilon - 2$ $L_1$ blocks in $\mathcal{P}_\psi$ can be used to create the block $L_\epsilon$ as explained above. Due to dimensionality reasons, there exists also a $L^T_\nu$ in $\mathcal{P}_\phi$ and we can use the $M^1(0)$ block together with $\nu$ $L_1$ blocks to create it. We are left with $m - 3 - (\epsilon-2) - \nu$ $L_1$ blocks which are used to create the remaining arbitrary $[m - (\epsilon + \nu + 1) ]\times[ m - (\epsilon + \nu + 1)]$ sized matrix pencil, as explained above.
Let us now consider case (ii). We keep one of the $L_1$ blocks in $\mathcal{P}_\psi$ unchanged, as this block is also present in $\mathcal{P}_\phi$. Again, due to dimensionality reasons, there exists a $L^T_\nu$ block in $\mathcal{P}_\phi$, which can be created using $L_2$ together with $\nu - 1$ $L_1$ blocks. We are left with $m - 3 -\nu$ $L_1$ blocks and the $M^1(0)$ block which are used to create the remaining arbitrary $[m - (2 + \nu)] \times [m - (2 + \nu)]$ sized matrix pencil, as explained above in the $[2 \times m \times( 2m-1)]$ scenario.
In the remaining case, case (iii), no $L$ block and hence (as $\mathcal{P}_\phi$ is square) also no $L^T$ block is present in $\mathcal{P}_\phi$, i.e., it only remains to be shown that $\mathcal{P}_\psi$ can be used to reach an arbitrary $J$ block of size $m \times m$. 
Due to the discussion above Theorem \ref{theo:fractional}, we only have to consider finite eigenvalues. Let us show that with one exception (we elaborate on that below), $\mathcal{P}_\psi$ can indeed be used to reach an arbitrary $J$ block of size $m \times m$.
First, Alice can transform $M^1(0)$ to $M^1(x)$ for an arbitrary $x$. The $L_2$ block can either be used to create $M^2(x)$ or $M^1(x_1) \oplus M^1(x_2)$, where $x_1 \neq x_2$ as proven in Lemma \ref{lemmanulltogeneric} and right below it. Together with $M^1(0)$, $L_2$ can also be used to create $M^3(x)$. As explained above, the remaining $L_1$ blocks can be used to create the remaining part of the $J$ block by either creating new $M^1(x)$ blocks, or by increasing the size signature of existing ones.
Let us note here, that with the method explained here, it is possible to reach any $J$, except pencils of the form $\mathcal{P}_\phi = m M^1(x)$. The reason for that is that $L_2$ cannot be used to create $M^1(x) \oplus M^1(x)$. However, this is not a problem as this matrix pencil corresponds to a state where Alice is not entangled with the other parties. 
This completes the proof.
\end{proof}

Let us remark here that Theorem \ref{theo:mmcr} does not hold for $m=3$, as
$
 \begin{pmatrix}
 \lambda & \mu & \cdot & \cdot \\
\cdot & \lambda & \mu & \cdot \\
 \cdot & \cdot & \cdot &  \lambda\\
\end{pmatrix}
$
can be used to reach any $3 \times 3$ matrix pencil except
$
 \begin{pmatrix}
 \lambda & \mu & \cdot  \\
\cdot & \cdot & \mu  \\
 \cdot & \cdot & \lambda \\
\end{pmatrix}
$. Moreover, no other $2 \times 3 \times 4$ state is a common resource for $2 \times 3 \times 3$, which can be easily verified by considering the full SLOCC hierarchy up to $2 \times 3 \times 6$ states which has been derived in Ref. \cite{ChMi10}. 
However, a common resource for $2 \times 3 \times 3$ states can be found in $\mathbb{C}^{2} \otimes \mathbb{C}^{3} \otimes \mathbb{C}^{5}$ instead. For example, the $2 \times 3 \times 5$ state corresponding to the matrix pencil$
 \begin{pmatrix}
 \lambda & \mu & \cdot & \cdot & \cdot \\
 \cdot & \cdot & \lambda & \mu & \cdot \\
 \cdot & \cdot & \cdot &\cdot &  \lambda\\
\end{pmatrix}
$ can be used to reach any $2 \times 3 \times 3$ state. 

We are going to show next that the resource state introduced in Theorem \ref{theo:mmcr} is an optimal resource for $\mathbb{C}^2\otimes \mathbb{C}^m\otimes \mathbb{C}^{m}$ states.

\begin{theorem}
\label{theo:mmocr}
The common resource state given in Theorem \ref{theo:mmcr} is optimal, i.e., no common resource state which reaches any $2\times m\times m$ state for $m \geq 4$ exists in $\mathbb{C}^2\otimes \mathbb{C}^m\otimes \mathbb{C}^{2m-3}$ or lower dimensions.  
\end{theorem}

\begin{proof}
To prove this theorem, we consider the full set of matrix pencils $\mathcal{P}_\psi$ corresponding to $2 \times m \times (2m-3)$ states $\ket{\psi}$ in SKCF. Due to the discussion above Theorem \ref{theo:fractional}, we only have to consider finite eigenvalues. Then, we successively eliminate classes of matrix pencils from this set until there are no candidates for a common resource left and we hence prove the theorem.

First, we exclude matrix pencils $\mathcal{P}_\psi$ that contain at least one $L^T$ block, as for such matrix pencils we have $D_m = 0$ (see proof of Lemma \ref{lemma:dmequalone}). Note that $D_m = 0$ is equivalent to the fact that there does not exist $m$ linearly independent columns in $\mathcal{P}_\psi$. As the columns of any matrix pencil $\mathcal{P}_\phi$ to which $\mathcal{P}_\psi$ can be transformed are linear superpositions of the columns of $\mathcal{P}_\psi$, it follows that states corresponding to such $\mathcal{P}_\psi$ can never reach states corresponding to matrix pencils for which $D_m \neq 0$. Thus, such matrix pencils cannot be a common resource to all $2\times m\times m$ states.

Let us now exclude matrix pencils $\mathcal{P}_\psi$ with at least two distinct eigenvalues $x_1$ and $x_2$. For any matrix pencil $\mathcal{P}_\phi$ that can be reached from such a $\mathcal{P}_\psi$, it holds that either $D_m=0$, or both $\mu x_1 + \lambda$ and $\mu x_2 + \lambda$ divide $D_m$. This can be seen as follows. If $\mathcal{P}_\psi$ has at least two distinct eigenvalues $x_1$ and $x_2$, then the pencil $\mathcal{P}_\psi$ (in KCF) contains at least one row, $k$, ($l$) that vanishes everywhere except for a single entry $\mu x_1 + \lambda$ ($\mu x_2 + \lambda$), respectively. 
Therefore, it holds that the entries of the $k$th ($l$th) row of any matrix pencil $\mathcal{P}_\phi$, to which $\mathcal{P}_\psi$ can be transformed, are multiples of $\mu x_1 + \lambda$ ($\mu x_2 + \lambda$), respectively. Hence, any $m$-minor of $\mathcal{P}_\phi$ either vanishes or is divisible by both $\mu x_1 + \lambda$ and $\mu x_2 + \lambda$. Hence, this also holds for $D_m$. This, however, implies that states corresponding to matrix pencils with an eigenvalue of algebraic multiplicity $m$, e.g., $\mathcal{P}_\phi = M^m(x)$, cannot be reached, as for such matrix pencils there cannot exist two distinct roots of $D_m$.

Let us now consider matrix pencils $\mathcal{P}_\psi$ that have an eigenvalue $x$ with an algebraic multiplicity of at least two (but arbitrary size signature), i.e., $\mathcal{P}_\psi$ either contains at least one $M^{e}(x)$ block with $e \geq 2$ or at least two $M^{1}(x)$ blocks, and show that they can be excluded. 
In the following, we will show that for any matrix pencil $\mathcal{P}_\phi$ that can be reached from $\mathcal{P}_\psi$, it holds that either $D_m=0$, or $(\mu x + \lambda)^2$ divides $D_m$. From that it follows that states corresponding to matrix pencils with $m$ distinct eigenvalues cannot be reached.
Let us consider $\mathcal{P}_\psi$ (here, w.l.o.g. we assume that $M^{e_j}(x)$ are the last blocks and $x=0$ to simplify notation) and the most general matrix pencil $\mathcal{P}_\phi$, which can be reached from $\mathcal{P}_\psi$, which can be written as
\begin{equation}
\begin{tikzpicture}[baseline=(current  bounding  box.center),
style1/.style={
  matrix of math nodes,
  every node/.append style={text width=#1,align=center,minimum height=3ex},
  nodes in empty cells,
  left delimiter=(,
  right delimiter=),
  },
]
\matrix[style1=0.25cm] (1mat)
{
  & & & & & & \\
  & & & & & & \\
  & & & & & & \\
  & & & & & \lambda & w \vphantom{\lambda}\\
  \vphantom{\lambda} & & & & & \vphantom{\lambda}& \lambda \\
};

\draw[solid]
  (1mat-1-1.north west) -- (1mat-1-6.north east);
\draw[solid]
  (1mat-3-1.south west) -- (1mat-3-6.south east);
\draw[solid]
  (1mat-1-1.north west) -- (1mat-3-1.south west);
\draw[solid]
  (1mat-1-6.north east) -- (1mat-3-6.south east);
  
\draw[solid]
  (1mat-4-6.north west) -- (1mat-4-7.north east);
\draw[solid]
  (1mat-5-6.south west) -- (1mat-5-7.south east);
\draw[solid]
  (1mat-4-6.north west) -- (1mat-5-6.south west);
\draw[solid]
  (1mat-4-7.north east) -- (1mat-5-7.south east);

\node[font=\huge] 
  at ([xshift=8pt,yshift=0pt]1mat-2-3) {$\cdots$}; 

\draw[decoration={brace,raise=12pt},decorate]
  (1mat-1-7.north east) -- 
  node[right=15pt] {$m$} 
  (1mat-5-7.south east); 
\draw[decoration={brace,mirror,raise=5pt},decorate]
  (1mat-5-1.south west) -- 
  node[below=7pt] {$2m-3$} 
  (1mat-5-7.south east);  

\node at ([xshift=-20pt,yshift=-1.2pt]1mat.west) {$\mathcal{P}_\psi =$};  
\end{tikzpicture}, \quad
\begin{tikzpicture}[baseline=(current  bounding  box.center),
style1/.style={
  matrix of math nodes,
  every node/.append style={text width=#1,align=center,minimum height=3ex},
  nodes in empty cells,
  left delimiter=(,
  right delimiter=),
  },
]
\matrix[style1=1.7cm] (1mat)
{
  \hphantom{a_m \lambda + b_m w }& & & \hphantom{a_m \lambda + b_m w }\\
  & & & \\
  \hphantom{a_m \lambda + b_m w }& & & \hphantom{a_m \lambda + b_m w }\\
  a_1 \lambda + b_1 w & a_2 \lambda + b_2 w & \cdots & a_m \lambda + b_m w \\
  b_1 \lambda & b_2 \lambda & \cdots & b_m \lambda \\
};

\draw[solid]
  (1mat-1-1.north west) -- (1mat-1-4.north east);
\draw[solid]
  (1mat-3-1.south west) -- (1mat-3-4.south east);
\draw[solid]
  (1mat-1-1.north west) -- (1mat-3-1.south west);
\draw[solid]
  (1mat-1-4.north east) -- (1mat-3-4.south east);
  
\node[font=\huge] 
  at ([xshift=18pt,yshift=0pt]1mat-2-2) {$\cdots$}; 
  
\draw[decoration={brace,raise=12pt},decorate]
  (1mat-1-4.north east) -- 
  node[right=15pt] {$m$} 
  (1mat-5-4.south east); 
\draw[decoration={brace,mirror,raise=5pt},decorate]
  (1mat-5-1.south west) -- 
  node[below=7pt] {$m$} 
  (1mat-5-4.south east);

\node at ([xshift=-20pt,yshift=-1.2pt]1mat.west) {$\mathcal{P}_\phi =$};  
\end{tikzpicture}.
\end{equation}
Here, $w$ either equals $\mu$ or vanishes.
Let us now calculate $D_m$ of $\mathcal{P}_\phi$. As we will see it suffices to consider the last two rows to verify the claim.
\begin{align}
D_m &= \det{\mathcal{P}_\phi} = \sum_{\sigma \in S_m} \operatorname{sgn}(\sigma) \prod_{i=1}^{m} {\mathcal{P}_\phi}_{i,\sigma_i} = \sum_{\sigma \in S_m} \operatorname{sgn}(\sigma) \prod_{i=1}^{m-2} {\mathcal{P}_\phi}_{i,\sigma_i} (a_{\sigma_{m-1}} \lambda + b_{\sigma_{m-1}} w) b_{\sigma_m} \lambda  \nonumber \\
&=\frac{1}{2} \sum_{\sigma \in S_m} \left[ \operatorname{sgn}(\sigma) \prod_{i=1}^{m-2} {\mathcal{P}_\phi}_{i,\sigma_i} (a_{\sigma_{m-1}} \lambda + b_{\sigma_{m-1}} w) b_{\sigma_m} \lambda \right. \nonumber \\
 &\qquad \left. + \operatorname{sgn}(\sigma \circ (m,m-1)) \prod_{i=1}^{m-2} {\mathcal{P}_\phi}_{i,\sigma \circ (m, m-1)_{i}} (a_{\sigma \circ (m, m-1)_{m-1}} \lambda + b_{\sigma \circ (m, m-1)_{m-1}} w) b_{\sigma \circ  (m, m-1)_{m}} \lambda \right] \nonumber\\
&=\frac{1}{2} \sum_{\sigma \in S_m} \left[ \operatorname{sgn}(\sigma) \prod_{i=1}^{m-2} {\mathcal{P}_\phi}_{i,\sigma_i} (a_{\sigma_{m-1}} \lambda + b_{\sigma_{m-1}} w) b_{\sigma_m} \lambda   - \operatorname{sgn}( \sigma) \prod_{i=1}^{m-2} {\mathcal{P}_\phi}_{i, \sigma_i} (a_{\sigma_{m}} \lambda + b_{ \sigma_{m}} w) b_{ \sigma_{m-1}} \lambda \right] \nonumber\\
&=\frac{\lambda^2}{2} \sum_{\sigma \in S_m} \left[ \operatorname{sgn}(\sigma) \prod_{i=1}^{m-2} {\mathcal{P}_\phi}_{i,\sigma_i} (a_{\sigma_{m-1}}  b_{\sigma_{m}}  - a_{\sigma_{m}}  b_{\sigma_{m-1}})  \right]  ,
\end{align}
where $S_m$ denotes the symmetric group of order $m$, $(i, j)$ denotes transposition of $i$ and $j$ and $\sigma \circ \tau$ denotes composition of $\sigma$ and $\tau$. As can be easily seen, $D_m$ either vanishes, or it is divisible by $\lambda^2$. Thus, $\mathcal{P}_\phi$ with $m$ distinct divisors cannot be reached and hence, such pencils $\mathcal{P}_\psi$ cannot be a common resource. 

We have hence reduced the considered set of matrix pencils $\mathcal{P}_\psi$ to candidates that contain no $L^T$ block and either no $J$ block or a $J$ block of size $1 \times 1$. Due to dimensionality reasons such matrix pencils either contain at least one $L_\epsilon$ block with $\epsilon \geq 3$ or $L_2 \oplus L_2$. In the following, we will complete the proof by showing that neither of these two classes of matrix pencils serve as common resource.

Let us first consider matrix pencils $\mathcal{P}_\psi$ that contain an $L_\epsilon$ block with $\epsilon \geq 3$, which we assume to be the first block w.l.o.g.. We will now show that for any matrix pencil $\mathcal{P}_\phi$ that can be reached from $\mathcal{P}_\psi$, it holds that either $D_m = 0$ or $D_2 = 1$, which implies that matrix pencils of the form $(m-1) M^1(x_1) \oplus M^1(x_2)$, where $x_1 \neq x_2$ \footnote{Recall that Alice would not be entangled with the other parties if $x_1 = x_2$}, cannot be reached as for such pencils $D_m = (x_1 \mu + \lambda)^{m-1} (x_2 \mu + \lambda) \neq 0$ and $D_2 = x_1 \mu + \lambda \neq 1$.
To this end, we consider the first three rows of a reachable matrix pencil $\mathcal{P}_\phi$, which will suffice to verify the claim. If there are less than three columns that are linearly independent in the first three rows, then $D_m = 0$, as in this case there cannot exist $m$ linearly independent columns in $\mathcal{P}_\phi$. Otherwise, let us consider the $3 \times 3$ submatrix of $\mathcal{P}_\phi$ obtained by considering the first three rows of three such linearly independent columns. We find that the greatest common divisor of the $2$-minors of this matrix equals 1 which implies that $D_2 =1$ as any 2-minor of this submatrix is also a $2$-minor of $\mathcal{P}_\phi$. To this end, note that any column in this submatrix is a superposition of the first four columns of $\mathcal{P}_\psi$ truncated to the first three rows (as they are the only columns which are non-vanishing in this subspace), i.e., a superposition of the colums of $L_3$ (note that we always consider the first three rows of $\mathcal{P}_
\psi$, even if the $L_\epsilon$ block in $\mathcal{P}_\psi$ is not $L_3$, but $L_\epsilon$ with $\epsilon > 3$ instead). This submatrix 
is strictly equivalent to a matrix that is obtained by adding the $i$th column of $L_3$ to the other columns with arbitrary coefficients, and deleting the $i$th column afterwards for some $i \in \{1,2,3,4\}$. It can be easily verified that for any $i$ and any choice of coefficients, there exists one $2$-minor proportional to $\lambda^2$ and one $2$-minor proportional to $\mu^2$ and hence $D_2 = 1$. This proves that pencils containing $L_\epsilon$ for $\epsilon \geq 3$ cannot be a common resource.

Let us finally consider matrix pencils $\mathcal{P}_\psi$ that contain $L_2 \oplus L_2$, which we assume to be the first blocks w.l.o.g.. Similarly as above, we will show that for any matrix pencil $\mathcal{P}_\phi$, that is reachable from such a $\mathcal{P}_\psi$, it holds that either $D_m = 0$ or $D_2 = 1$ and hence such matrix pencils $\mathcal{P}_\psi$ cannot be a common rersource as explained above. 
To this end, we consider the first four rows of a reachable matrix pencil $\mathcal{P}_\phi$, which will suffice to verify the claim. If there are less than four columns that are linearly independent in the first four rows, then $D_m = 0$. Otherwise, let us consider the $4 \times 4$ submatrix of $\mathcal{P}_\phi$ obtained by considering the first four rows of four such linearly independent columns. Similarly as above, we show that the greatest common divisor of the $2$-minors of this matrix equals $1$, which implies that $D_2 =1$.
Similarly as above, this submatrix is strictly equivalent to a matrix that is obtained by adding the $i$th and the $j$th column of $L_2 \oplus L_2$ to the other columns with arbitrary coefficients, and deleting the $i$th and the $j$th column afterwards for some $i,j \in \{1,2,3,4\}$, $i\neq j$. It can be easily verified that for any $i$, $j$, and any choice of coefficients, there exists one $2$-minor proportional to $\lambda^2$ and one $2$-minor proportional to $\mu^2$ and hence $D_2 = 1$, which proves that there does not exist any $2\times m\times (2m-3)$ state that is common resource to all $2\times m\times m$ states.

It follows also that no common resource state, $\ket{\psi}$, can exists in lower dimensions, as otherwise there would exist at least one $2\times m\times (2m-3)$ state that can reach $\ket{\psi}$ and thus would be a common resource for $2 \times m \times m$ states, which would contradict our finding above.
This completes the proof.
\end{proof}

As an illustrative example we present the optimal common resource state for all $2\times 4\times 4$ states in Appendix \ref{sec:crex}.

Let us now consider the optimal common resource state for $2 \times m \times n$ states where $n > m$. We will see that in contrast to the $2 \times m \times m$ scenario, where we found an optimal resource state in $\mathbb{C}^2\otimes \mathbb{C}^m\otimes \mathbb{C}^{2m-2}$, here, the trivial $2 \times m \times 2m$ state that can be used to perform teleportation is optimal. In other words, no common resource for $2 \times m \times n$ states exists in $\mathbb{C}^2\otimes \mathbb{C}^m\otimes \mathbb{C}^{2m-1}$ or lower dimensions.

\begin{theorem}
\label{theo:mnocr}
A common resource state which reaches any $2\times m\times n$ state for $m \geq 2$, $n > m$ exists in $\mathbb{C}^2\otimes \mathbb{C}^m\otimes \mathbb{C}^{2m}$. It is optimal, i.e., no common resource exists in $\mathbb{C}^2\otimes \mathbb{C}^m\otimes \mathbb{C}^{2m-1}$ or lower dimensions.  
\end{theorem}

\begin{proof}
As discussed above, the $2\times m\times 2m$ state $\ket{\psi}$ corresponding to $\mathcal{P}_\psi = m L_1$ can be transformed to any other state of lower dimension, in particular to any $2 \times m \times n$ state. This can even be achieved deterministically as Claire can use $\ket{\psi}$ to perform teleportation.

Let us now prove that this state is optimal. To this end, we consider the set of matrix pencils $\mathcal{P}_\psi$ corresponding to $2 \times m \times (2m-1)$ states $\ket{\psi}$ in SKCF and prove the statement in two steps. First, we show that a common resource $\mathcal{P}_\psi$ may not contain a $J$-block. Due to dimensionality reasons, only one candidate remains, which we prove not to be a common resource in a second step.

Let us consider matrix pencils $\mathcal{P}_\psi$ with at least one eigenvalue $x$. For any matrix pencil $\mathcal{P}_\phi$ that can be reached from such a $\mathcal{P}_\psi$ it holds that either $D_m=0$, or $\mu x + \lambda$ divides $D_m$, which can be proven utilizing the same arguments as in the proof of Theorem \ref{theo:mmocr}. In particular, we use that the entries of some row of any reachable $\mathcal{P}_\phi$ are multiples of $\mu x + \lambda$ only. Thus, states corresponding to matrix pencils that are a sum of right nullspace blocks only, $\mathcal{P}_\phi = \bigoplus_i L_{\epsilon_i}$, cannot be reached, as for such $\mathcal{P}_\phi$, it holds that $D_m=1$ due to Lemma \ref{lemma:dmequalone}.

Let us now consider matrix pencils $\mathcal{P}_\psi$ without eigenvalues, i.e., without a $J$ block. Due to dimensionality reasons there exists only one such matrix pencil, $\mathcal{P}_\psi = (m-2) L_1 \oplus L_2$. We will now complete the proof by showing that the $2 \times m \times (m+1)$ state corresponding to $\mathcal{P}_\phi = L_1 \oplus (m-1) M^1(0)$ cannot be reached from $\ket{\psi}$, as this implies that, moreover, there exists at least one $2 \times m \times n$ state that cannot be reached from $\ket{\psi}$. Thus, $\ket{\psi}$ cannot be a common resource for $2 \times m \times n$ states.
Instead of considering the matrix pencils in KCF, we consider the following strictly equivalent pencils for convenience

\begin{equation}
\begin{tikzpicture}[baseline=(current  bounding  box.center),
style1/.style={
  matrix of math nodes,
  every node/.append style={text width=#1,align=center,minimum height=1.0ex},
  nodes in empty cells,
  left delimiter=(,
  right delimiter=),
  },
]
\matrix[style1=0.25cm] (1mat)
{
  \lambda & & &  \mu& & & \vphantom{\mu}\\
  & \lambda& & &   \mu & &\\
  & & \ddots & & &   \ddots & \\
  \vphantom{\mu}& & & \lambda \vphantom{\mu} & \vphantom{\mu} & & \mu\\
};
\draw[solid]
  (1mat-1-4.north east) -- (1mat-4-4.south east);

\draw[decoration={brace,raise=12pt},decorate]
  (1mat-1-7.north east) -- 
  node[right=15pt] {$m$} 
  (1mat-4-7.south east); 
\draw[decoration={brace,mirror,raise=5pt},decorate]
  (1mat-4-1.south west) -- 
  node[below=7pt] {$m$} 
  (1mat-4-4.south east);  
 \draw[decoration={brace,mirror,raise=5pt},decorate]
  (1mat-4-5.south west) -- 
  node[below=7pt] {$m-1$} 
  (1mat-4-7.south east);  

\node at ([xshift=-20pt,yshift=-1.2pt]1mat.west) {$\mathcal{P}_\psi =$};  
\end{tikzpicture}, \quad
\begin{tikzpicture}[baseline=(current  bounding  box.center),
style1/.style={
  matrix of math nodes,
  every node/.append style={text width=#1,align=center,minimum height=3ex},
  nodes in empty cells,
  left delimiter=(,
  right delimiter=),
  },
]
\matrix[style1=0.25cm] (1mat)
{
  \lambda & & &  \vphantom{\mu}&\mu\\
  & \lambda& & & 0 \\
  & & \ddots & &  \vdots\\
  \vphantom{\mu}& & & \lambda \vphantom{\mu} & 0 \vphantom{\mu} \\
};
\draw[solid]
  (1mat-1-4.north east) -- (1mat-4-4.south east);

\draw[decoration={brace,raise=12pt},decorate]
  (1mat-1-5.north east) -- 
  node[right=15pt] {$m$} 
  (1mat-4-5.south east); 
\draw[decoration={brace,mirror,raise=5pt},decorate]
  (1mat-4-1.south west) -- 
  node[below=7pt] {$m$} 
  (1mat-4-4.south east);  
 \draw[decoration={brace,mirror,raise=5pt},decorate]
  (1mat-4-5.south west) -- 
  node[below=7pt] {$1$} 
  (1mat-4-5.south east);  

\node at ([xshift=-20pt,yshift=-1.2pt]1mat.west) {$\mathcal{P}_\phi =$};  
\end{tikzpicture}.
\end{equation}
Let us now assume that there exist operators $B$ and $C$ such that $B \mathcal{P}_\psi C^T = \mathcal{P}_\phi$ and see that this leads to a contradiction. As $B \mathcal{P}_\psi C^T = \mathcal{P}_\phi$ has to hold for all $\lambda$ and $\mu$, equivalently, the following two equalities hold
\begin{equation}
\label{eq:lambdacoeff}
\begin{tikzpicture}[baseline=(current  bounding  box.center),
style1/.style={
  matrix of math nodes,
  every node/.append style={text width=#1,align=center,minimum height=1.0ex},
  nodes in empty cells,
  left delimiter=(,
  right delimiter=),
  },
]
\matrix[style1=0.25cm] (1mat)
{
  & & \vphantom{1}& & &\vphantom{1}\\
  & & & & &\\
  & & \vphantom{1}& & &\vphantom{1}\\
};

\draw[solid]
  (1mat-1-3.north east) -- (1mat-3-3.south east);

\draw[decoration={brace,raise=12pt},decorate]
  (1mat-1-6.north east) -- 
  node[right=15pt] {$m$} 
  (1mat-3-6.south east); 
\draw[decoration={brace,mirror,raise=5pt},decorate]
  (1mat-3-1.south west) -- 
  node[below=7pt] {$m$} 
  (1mat-3-3.south east);  
 \draw[decoration={brace,mirror,raise=5pt},decorate]
  (1mat-3-4.south west) -- 
  node[below=7pt] {$m-1$} 
  (1mat-3-6.south east);  

\node[font=\huge] 
  at ([xshift=0pt,yshift=0pt]1mat-2-2) {$\identity$};
 \node[font=\huge] 
  at ([xshift=0pt,yshift=0pt]1mat-2-5) {$0$};

\node at ([xshift=-20pt,yshift=-1.2pt]1mat.west) {$B$};  
\node at ([xshift=+40pt,yshift=-1.2pt]1mat.east) {$C^T = $};  
\end{tikzpicture} 
\begin{tikzpicture}[baseline=(current  bounding  box.center),
style1/.style={
  matrix of math nodes,
  every node/.append style={text width=#1,align=center,minimum height=1ex},
  nodes in empty cells,
  left delimiter=(,
  right delimiter=),
  },
]
\matrix[style1=0.25cm] (1mat)
{
  & & \vphantom{1}& \vphantom{1}\\
  & & &\\
  & & \vphantom{1}& \vphantom{1}\\
};

\draw[solid]
  (1mat-1-3.north east) -- (1mat-3-3.south east);

\node[font=\huge] 
  at ([xshift=0pt,yshift=0pt]1mat-2-2) {$\identity$};
  \node[font=\large] 
  at ([xshift=0pt,yshift=0pt]1mat-2-4) {$\vec{0}$};

\draw[decoration={brace,raise=12pt},decorate]
  (1mat-1-4.north east) -- 
  node[right=15pt] {$m$} 
  (1mat-3-4.south east); 
\draw[decoration={brace,mirror,raise=5pt},decorate]
  (1mat-3-1.south west) -- 
  node[below=7pt] {$m$} 
  (1mat-3-3.south east);  
\draw[decoration={brace,mirror,raise=5pt},decorate]
  (1mat-3-4.south west) -- 
  node[below=7pt] {$1$} 
  (1mat-3-4.south east);  

\end{tikzpicture},
\end{equation}
\begin{equation}
\label{eq:mucoeff}
\begin{tikzpicture}[baseline=(current  bounding  box.center),
style1/.style={
  matrix of math nodes,
  every node/.append style={text width=#1,align=center,minimum height=1.0ex},
  nodes in empty cells,
  left delimiter=(,
  right delimiter=),
  },
]
\matrix[style1=0.25cm] (1mat)
{
  & & \vphantom{1}& & &\vphantom{1}\\
  & & & & &\\
  & & \vphantom{1}& & &\vphantom{1}\\
};

\draw[solid]
  (1mat-1-3.north east) -- (1mat-3-3.south east);

\draw[decoration={brace,raise=12pt},decorate]
  (1mat-1-6.north east) -- 
  node[right=15pt] {$m$} 
  (1mat-3-6.south east); 
\draw[decoration={brace,mirror,raise=5pt},decorate]
  (1mat-3-1.south west) -- 
  node[below=7pt] {$m-1$} 
  (1mat-3-3.south east);  
 \draw[decoration={brace,mirror,raise=5pt},decorate]
  (1mat-3-4.south west) -- 
  node[below=7pt] {$m$} 
  (1mat-3-6.south east);  

\node[font=\huge] 
  at ([xshift=0pt,yshift=0pt]1mat-2-2) {$0$};
 \node[font=\huge] 
  at ([xshift=0pt,yshift=0pt]1mat-2-5) {$\identity$};

\node at ([xshift=-20pt,yshift=-1.2pt]1mat.west) {$B$};  
\node at ([xshift=+40pt,yshift=-1.2pt]1mat.east) {$C^T = $};  
\end{tikzpicture} 
\begin{tikzpicture}[baseline=(current  bounding  box.center),
style1/.style={
  matrix of math nodes,
  every node/.append style={text width=#1,align=center,minimum height=1ex},
  nodes in empty cells,
  left delimiter=(,
  right delimiter=),
  },
]
\matrix[style1=0.25cm] (1mat)
{
  & & \vphantom{1}& \vphantom{1}\\
  & & &\\
  & & \vphantom{1}& \vphantom{1}\\
};

\draw[solid]
  (1mat-1-3.north east) -- (1mat-3-3.south east);

\node[font=\huge] 
  at ([xshift=0pt,yshift=0pt]1mat-2-2) {$0$};
  \node[font=\large] 
  at ([xshift=0pt,yshift=0pt]1mat-2-4) {$\vec{e}_0$};

\draw[decoration={brace,raise=12pt},decorate]
  (1mat-1-4.north east) -- 
  node[right=15pt] {$m$} 
  (1mat-3-4.south east); 
\draw[decoration={brace,mirror,raise=5pt},decorate]
  (1mat-3-1.south west) -- 
  node[below=7pt] {$m$} 
  (1mat-3-3.south east);  
\draw[decoration={brace,mirror,raise=5pt},decorate]
  (1mat-3-4.south west) -- 
  node[below=7pt] {$1$} 
  (1mat-3-4.south east);  

\end{tikzpicture}.
\end{equation}
Note that Eq. (\ref{eq:lambdacoeff}) implies that $B$ is invertible as the matrix on the right hand side of Eq. (\ref{eq:lambdacoeff}) is of rank $m$. Furthermore, Eq. (\ref{eq:lambdacoeff}) and Eq. (\ref{eq:mucoeff}) imply that $C^T$ has both of the following forms respectively,
\begin{equation}
\begin{tikzpicture}[baseline=(current  bounding  box.center),
style1/.style={
  matrix of math nodes,
  every node/.append style={text width=#1,align=center,minimum height=3.0ex},
  nodes in empty cells,
  left delimiter=(,
  right delimiter=),
  },
]
\matrix[style1=0.25cm] (1mat)
{
  & & & &\\
  & & & &\\
  & & & &\\
  & & & &\\
  & & & &\\
  & & & &\\
  & & & &\\
};

\draw[solid]
  (1mat-1-4.north east) -- (1mat-4-4.south east);
\draw[solid]
  (1mat-4-1.south west) -- (1mat-4-5.south east);

\node[font=\large] 
  at ([xshift=5pt,yshift=-5pt]1mat-2-2) {$B^{-1}$};
\node[font=\large] 
  at ([xshift=0pt,yshift=-5pt]1mat-2-5) {$\vec{0}$};
\node[font=\large] 
  at ([xshift=0pt,yshift=0pt]1mat-6-3) {$\text{arbitrary}$}; 

\draw[decoration={brace,raise=12pt},decorate]
  (1mat-1-5.north east) -- 
  node[right=15pt] {$m$} 
  (1mat-4-5.south east); 
 \draw[decoration={brace,raise=12pt},decorate]
  (1mat-5-5.north east) -- 
  node[right=15pt] {$m-1$} 
  (1mat-7-5.south east); 
\draw[decoration={brace,mirror,raise=5pt},decorate]
  (1mat-7-1.south west) -- 
  node[below=7pt] {$m$} 
  (1mat-7-4.south east);  
 \draw[decoration={brace,mirror,raise=5pt},decorate]
  (1mat-7-5.south west) -- 
  node[below=7pt] {$1$} 
  (1mat-7-5.south east);  

\node at ([xshift=-25pt,yshift=-1.2pt]1mat.west) {$C^T = $};  
\end{tikzpicture} 
\begin{tikzpicture}[baseline=(current  bounding  box.center),
style1/.style={
  matrix of math nodes,
  every node/.append style={text width=#1,align=center,minimum height=3.0ex},
  nodes in empty cells,
  left delimiter=(,
  right delimiter=),
  },
]
\matrix[style1=0.25cm] (1mat)
{
  & & & &\\
  & & & &\\
  & & & &\\
  & & &  &\\
  & & & &\\
  & & & &\\
  & & & &\\
};

\draw[solid]
  (1mat-4-4.north east) -- (1mat-7-4.south east);
\draw[solid]
  (1mat-3-1.south west) -- (1mat-3-5.south east);

\node[font=\huge] 
  at ([xshift=5pt,yshift=-5pt]1mat-5-2) {$0$}; 
\node[font=\large] 
  at ([xshift=0pt,yshift=0pt]1mat-2-3) {$\text{arbitrary}$}; 

\draw[decoration={brace,raise=12pt},decorate]
  (1mat-1-5.north east) -- 
  node[right=15pt] {$m-1$} 
  (1mat-3-5.south east); 
 \draw[decoration={brace,raise=12pt},decorate]
  (1mat-4-5.north east) -- 
  node[right=15pt] {$m$} 
  (1mat-7-5.south east); 
\draw[decoration={brace,mirror,raise=5pt},decorate]
  (1mat-7-1.south west) -- 
  node[below=7pt] {$m$} 
  (1mat-7-4.south east);  
 \draw[decoration={brace,mirror,raise=5pt},decorate]
  (1mat-7-5.south west) -- 
  node[below=7pt] {$1$} 
  (1mat-7-5.south east);  

\node at ([xshift=-25pt,yshift=-1.2pt]1mat.west) {$, \ C^T = $};  
\end{tikzpicture}. 
\end{equation}
Now, note that Eq. (\theequation) implies that the last row of $B^{-1}$ vanishes, which contradicts the fact that $B$ and thus also $B^{-1}$ are invertible. This completes the proof.
\end{proof}

\section{Conclusions}

In conclusion, we have investigated the hierarchy of pure quantum states representing SLOCC classes  in $2 \times m \times n$-systems. We used matrix pencils and their Kronecker
normal form to  identify SLOCC classes and to study most general non-invertible transformations between 
them \cite{ChMi10}. This allowed us to find and parametrize generic SLOCC classes in every dimension and to find possible non-invertible transformations between them. Moreover, we identified resource states in higher dimensions, which can be  used to probabilistically generate all possible (including zero measure sets) SLOCC classes in a lower dimensional system. These results lead to a coarse graining of SLOCC classes. The identification of such resource states is 
also relevant from the point of view of state discrimination as will be explained in the following. It has 
been shown that if a state $\ket{\Phi}$ on a multipartite high-dimensional system 
is a resource state for all states in a lower-dimensional multipartite system, 
then its complex conjugate $\ket{\Phi^*}$ is a universal resource for unambiguous 
state discrimination in the lower-dimensional system \cite{bandhia}. More precisely,
any state out of a fixed set of linearly independent states $\{\ket{\psi_i}\}$ in the low-dimensional system can be correctly identified with non-vanishing probability of success (and a fail to do so can be detected), 
if the parties share in addition the resource $\ket{\Phi^*}$.

There are several directions in which our work may be generalized. First, it may 
be useful to find new polynomial invariants characterizing the SLOCC classes \cite{GoWa13,OsSi12}. These invariants can be used to construct entanglement measures such as the so-called three-tangle.
Moreover, it would be appealing to use the results obtained here to gain more insight into the entanglement properties of three-partite systems. Apart from the relation among the SLOCC classes which is induced via the hierarchy studied here, a more detailed investigation of entanglement seems possible. In particular, it would be desirable to study the derived resource states further and to understand their entanglement properties better. Furthermore, for some of the considered SLOCC classes it might prove promising to study LOCC transformations within them, as the symmetries of states in those classes indicate that the structure of possible LOCC transformations lies somewhere between the very simple bipartite LOCC structure and truly multipartite LOCC structure.  Finally, it would be very interesting to go beyond the somehow artificial restriction of $2 \times m \times n$-dimensional systems and consider general three-partite systems. This, however, probably requires a significant extension of the present theory of matrix pencils.

\begin{acknowledgments}
We thank Christina Ritz, Matthias Kleinmann, and Markus Grassl for discussions. We are very grateful to Stefan Johansson as he shared the  software package guptri, which amongst other things can be used to compute KCF of linear matrix pencils \cite{JoSite} with the algorithms presented in Ref. \cite{DeKa93} and \cite{DeKa93(2)}. M.G. acknowledges each and every  member  of B.K.'s group and Emanuele Grimaldo  in Innsbruck for their amazing hospitality and constant support for the three months of her visit and, additionally, M.G. would like to thanks B.K and O.G for letting her have this experience. M.H. and B.K. acknowledge financial support from the Austrian Science Fund 
(FWF) grants Y535-N16 and DK-ALM: W1259-N27. M.G. and O.G acknowledge financial
support from the DFG and the ERC (Consolidator Grant 683107/TempoQ). Additionally, 
M.G. would like to acknowledge funding from the Gesellschaft der Freunde und 
F\"orderer der Universit\"at Siegen.
\end{acknowledgments}

\appendix

\section{The union of the SLOCC classes of the set of states corresponding to generic pencils is of full measure}
\label{app:generic}

We show here that the SLOCC classes of states corresponding to generic matrix pencils are of full measure. More precisely, we restate and prove here Theorem \ref{theo:genstates}, which has been presented and discussed in the main text.

\begingroup
\def\thetheorem{\ref{theo:genstates}}
\begin{theorem}
The set of full rank states in $2\times m\times n$ belonging to a SLOCC class with a representative whose corresponding matrix pencil is generic, is of full measure.
\end{theorem}
\addtocounter{theorem}{-1}
\endgroup
\begin{proof}
Let us first consider the case $m=n$, then we consider the case $m \neq n$.
Recall that the generic matrix pencils of size $m \times m$ are those with $m$ distinct eigenvalues,  $\vec{x} \in \mathbb{C}^{m}$. Here, we will denote the corresponding SLOCC classes by $SLOCC_{\vec{x}}$. We will now show that the union of those SLOCC classes, $\bigcup_{\vec{x}} SLOCC_{\vec{x}}$, is of full measure.
Due to the definition of $SLOCC_{\vec{x}}$, for any state $\ket{\psi}$ in one of those SLOCC classes, there exist invertible operators $A$, $B$, and $C$ such that $(A \otimes B \otimes C) \ket{\psi} =   \left[ \ket{0}_A (R \otimes \identity) + \ket{1}_A (S \otimes \identity) \right] \ket{\phi^+_m}_{BC}$. Here, the matrix $R$ is invertible and the matrix $R^{-1} S$ has $m$ distinct eigenvalues. Recall that, alternatively, the operator $A$ could also be used to fix three of the eigenvalues to, e.g., $0$, $1$, and $\infty$. Introducing the two conditions:
\begin{enumerate}[(i)]
	\item $\det{R}\neq 0$,
	\item  all $m$ eigenvalues of $R^{-1}S$ are distinct, i.e., $\operatorname{char}(R^{-1}S)$ has no multiple roots, where $\operatorname{char}$ denotes the characteristic polynomial of a matrix,
\end{enumerate}
we  have
\begin{align}
	\bigcup_{\vec{x}} SLOCC_{\vec{x}} = &\left\{ \ket{\psi} : \exists A \in SL(2,\mathbb{C}), B \in SL(m,\mathbb{C}), C \in SL(n,\mathbb{C}) \text{ such that } (A \otimes B \otimes C) \ket{\psi} = \right. \\
	&\left.  \left[ \ket{0}_A (R \otimes \identity) + \ket{1}_A (S \otimes \identity) \right] \ket{\phi^+_m}_{BC}, \text{such that (i) and (ii) are fulfilled} \right\}.
\end{align}
To show that $\bigcup_{\vec{x}} SLOCC_{\vec{x}}$ is of full measure, we consider the sets
\begin{align}
F^B = \left\{ \ket{\psi} : \ket{\psi} =  \left[ \ket{0}_A (R \otimes \identity) + \ket{1}_A (S \otimes \identity) \right] \ket{\phi^+_m}_{BC}, \text{ such that (i) and (ii) are fulfilled} \right\},
\end{align}
where $\ket{0}$, $\ket{1}$, and $\ket{\phi^+_m}$ are defined with respect to an arbitrary local, orthonormal basis $B$, here. Obviously, $F^B \subseteq \bigcup_{\vec{x}} SLOCC_{\vec{x}}$ for any chosen basis $B$. We will now show that $F^B$ is of full measure for any $B$, which directly implies that $\bigcup_{\vec{x}} SLOCC_{\vec{x}}$ is of full measure. To this end, we consider $\mathcal{H} \setminus F^B$ and show that this set is of measure zero.
Note that
\begin{align}
\mathcal{H} \setminus F^B = \left\{ \ket{\psi} : \ket{\psi} =  \left[ \ket{0}_A (R \otimes \identity) + \ket{1}_A (S \otimes \identity) \right] \ket{\phi^+_m}_{BC}, \text{ such that either (i) or (ii) are not satisfied} \right\}
\end{align}
is a subset of $\mathcal{H}$ that is characterized by nontrivial polynomial equations in the coefficients of the state $\ket{\psi}$. This can be easily seen as follows. Obviously, the equation $\det{R}=0$ corresponds to a polynomial equation in the coefficients of the state. If $\det{R} \neq 0$, the fact that $\operatorname{char}(R^{-1}S)$ has at least one multiple root is equivalent to the condition that the discriminant $\Delta$ of $\det(S - \lambda R)$, a polynomial in $\lambda$, vanishes. The discriminant $\Delta$ of a polynomial $f$ is given as $\Delta(f) = \prod_{i < j} (\lambda_i - \lambda_j)^2$, where $\lambda_i$ are the roots of $f$. In fact, the discriminant is a polynomial in the coefficients of $f$ and hence, $\Delta(\det(S - \lambda R))$ is a polynomial of the coefficients of the state $\ket{\psi}$ \cite{GeKa94}. Finally, a subset of $\mathcal{H}$ which is characterized by nontrivial polynomial equations is indeed of measure zero. This completes the proof for $m=n$.

Let us now consider the case $n>m$. First, we will show that the union of SLOCC classes which correspond to matrix pencils that are a direct sum of right null space blocks only, is of full measure. Then, we show that the dimension of the SLOCC-orbit of the state corresponding to a generic matrix pencil is strictly larger than the orbit of the other representatives of SLOCC classes in this finite union. This will prove that the SLOCC class corresponding to a generic matrix pencil is of full measure.
Let us denote the above mentioned union of SLOCC classes as $\bigcup_{\vec{\epsilon}} SLOCC_{\vec{\epsilon}}$, where $\vec{\epsilon} \in \mathbb{N}^{n-m}$. The vector $\vec{\epsilon}$ corresponds to the minimal indices of the matrix pencil that is corresponding to $SLOCC_{\vec{\epsilon}}$, i.e., the sizes of the right nullspace blocks.
Using Lemma \ref{lemma:dmequalone} and the fact that if $D_m=1$ for a matrix pencil corresponding to some state $\ket{\psi}$, then $D_m=1$ for the matrix pencil corresponding to $A \otimes B \otimes C \ket{\psi}$ for any local invertible operators $A$, $B$, and $C$, we can write
\begin{align}
	\bigcup_{\vec{\epsilon}} SLOCC_{\vec{\epsilon}} = &\left\{ \ket{\psi} : \ket{\psi} = \left[ \ket{0}_A (R \otimes \identity) + \ket{1}_A (S \otimes \identity) \right] \ket{\phi^+_m}_{BC}, \text{ such that for } \mathcal{P}(R,S) \text{ it holds that } D_m=1 \right\}.
\end{align}
For a given pair of matrices $(R,S)$, let us denote the determinant of the $m \times m$ submatrix obtained by deleting all but the first $m$ columns of $\mu R + \lambda S$ by $d_1$ and the determinant of the $m \times m$ submatrix obtained by deleting all but the last $m$ columns of $\mu R + \lambda S$ by $d_2$.
To show that $\bigcup_{\vec{\epsilon}} SLOCC_{\vec{\epsilon}}$ is of full measure, we consider the sets
\begin{align}
	F^B = &\left\{ \ket{\psi} : \ket{\psi} = \left[ \ket{0}_A (R \otimes \identity) + \ket{1}_A (S \otimes \identity) \right] \ket{\phi^+_m}_{BC}, \text{ such that } \operatorname{gcd}(d_1,d_2)=1 \right\},
\end{align}
where $\ket{0}$, $\ket{1}$, and $\ket{\phi^+_m}$ are defined with respect to an arbitrary local, orthonormal basis $B$, here. Obviously, $F^B \subseteq \bigcup_{\vec{\epsilon}} SLOCC_{\vec{\epsilon}}$ for any chosen basis $B$, as $\operatorname{gcd}(d_1,d_2)=1$ is a sufficient condition for $D_m=1$. We will now show that $F^B$ is of full measure for any $B$, which directly implies that $\bigcup_{\vec{\epsilon}} SLOCC_{\vec{\epsilon}}$ is of full measure. To this end, we consider $\mathcal{H} \setminus F^B$ and show that this set is of measure zero.
Introducing the two conditions
\begin{enumerate}[(i)]
	\item $\operatorname{gcd}(d_1, d_2) \neq 1$ and neither $d_1$ nor $d_2$ vanish,
	\item $d_1$ vanishes or $d_2$ vanishes,
\end{enumerate}
we have
\begin{align}
	 \mathcal{H} \setminus F^B = &\left\{ \ket{\psi} : \ket{\psi} = \left[ \ket{0}_A (R \otimes \identity) + \ket{1}_A (S \otimes \identity) \right] \ket{\phi^+_m}_{BC}, \text{ such that either condition (i) or (ii) holds} \right\}.
\end{align}
We will now show that this set is characterized by nontrivial polynomial equations in the coefficients of the state $\ket{\psi}$. This can be seen as follows. Obviously, condition (ii) corresponds to a non-trivial polynomial equation in the coefficients of the state $\ket{\psi}$. In order to see that condition (i) also does, we consider the determinants $d_1$ and $d_2$.

Note that in case condition (i) holds, as explained before, they are non-vanishing homogenous polynomials in $\lambda$ and $\mu$ and uniquely factorize as
\begin{align}
	d_1 \propto (x_1 \mu + \lambda) \ldots (x_{m-t_1} \mu + \lambda) \mu^{t_1}\\
	d_2 \propto (y_1 \mu + \lambda) \ldots (y_{m-t_2} \mu + \lambda) \mu^{t_2}.
\end{align}
As $\operatorname{gcd}(d_1, d_2) \neq 1$, they must share at least one common factor, that is either a common factor $\mu$ or some $(x_k \mu + \lambda)$ in the factorization of $d_1$ and some $(y_l \mu + \lambda)$ in the factorization of $d_2$ with $x_k=y_l$. This is the case iff these polynomials share at least one common root $(\mu_0, \lambda_0)$ other than $(0,0)$. Let us consider the resultant $\operatorname{res}$ of $d_1$ and $d_2$. The resultant of two homogenous polynomials $f = f_0 \mu^m + f_1 \mu^{m-1} \lambda + \ldots + f_m \lambda^m$ and $g = g_0 \mu^m + g_1 \mu^{m-1} \lambda + \ldots + g_m \lambda^m$ is defined as
\begin{align}
\operatorname{res}(f,g) = \det \begin{pmatrix}\
f_m & f_{m-1} & \cdots & & f_0 & & &\\
 & f_m & f_{m-1} & \cdots & & f_0 & &\\
 & & \ddots & & & & \ddots &\\
 & & & f_m & f_{m-1} & \cdots & & f_0\\
g_m & g_{m-1} & \cdots & & g_0 & & &\\
 & g_m & g_{m-1} & \cdots & & g_0 & &\\
 & & \ddots & & & & \ddots &\\
 & & & g_m & g_{m-1} & \cdots & & g_0
\end{pmatrix}.
\end{align}
The fact that ${d}_1$ and ${d}_2$ have some common root other than $(0,0)$ is equivalent to the fact that their resultant, $\operatorname{res}(d_1, d_2)$, vanishes \cite{GeKa94}. Again, this translates to a polynomial equation in the coefficients of the state $\ket{\psi}$. We hence have that the set $\mathcal{H} \setminus F^B$ is completely characterized by nontrivial polynomial equations and thus is of measure zero.
Thus, we have that the set $\bigcup_{\vec{\epsilon}} SLOCC_{\vec{\epsilon}}$, which is the union of states of finitely many SLOCC classes, is of full measure. 
Recall that the generic matrix pencil is $\mathcal{P}(R,S)=(d-(m\ \operatorname{mod}\ d))L_{\left\lfloor m/d \right\rfloor}\oplus (m\ \operatorname{mod}\ d)L_{\left\lceil m/d\right\rceil }$ and the corresponding SLOCC class is $SLOCC_{\vec{\epsilon}'}$ with $\epsilon'_1 = \ldots = \epsilon'_{(d-(m\ \operatorname{mod}\ d))} = \left\lfloor m/d \right\rfloor$, and if $m\ \operatorname{mod}\ d \neq 0$, $\epsilon'_{(d-(m\ \operatorname{mod}\ d) + 1 )} = \ldots = \epsilon'_d = \left\lceil m/d\right\rceil$.
In the following, we will complete the proof by showing that $SLOCC_{\vec{\epsilon}}$ for any $\vec{\epsilon} \neq \vec{\epsilon}'$ is of measure zero and thus also the finite union $\bigcup_{\vec{\epsilon} \neq \vec{\epsilon}'} SLOCC_{\vec{\epsilon}}$ is of measure zero. From this it follows that $SLOCC_{\vec{\epsilon'}}$ is of full measure. Let us now consider the orbit $\mathcal{O}_{\vec{\epsilon}} = \{ \identity \otimes \tilde{B} \otimes \tilde{C} \ket{\psi}, \text{ where } \ket{\psi} \text{ is representative of } SLOCC_{\vec{\epsilon}}, \tilde{B} \in GL(m, \mathbb{C}), \tilde{C} \in GL(n,\mathbb{C}) \}$. In \cite{DeEd95}, the dimensions of these orbits $\operatorname{dim}\{\mathcal{O}_{\vec{\epsilon}}\}$ have been characterized. In particular, it has been shown that $\operatorname{dim}\{\mathcal{O}_{\vec{\epsilon}'}\} = 2 m n$ and $\operatorname{dim}\{\mathcal{O}_{\vec{\epsilon}}\} < 2 m n$ for any $\vec{\epsilon} \neq \vec{\epsilon}'$, where the complex dimension is counted. Now recall Lemma \ref{lemma:undoalice}, which states that for 
states $\ket{\psi}$ corresponding to matrix pencils that only contain $L_{\epsilon}$ and $L_{\nu}$ blocks, any action $A$ by the first party can be undone by the other parties. This implies that the orbits under $A\otimes B \otimes C$ coincide with the orbits $\identity \otimes B \otimes C$, i.e.,
$\{ A \otimes B \otimes C \ket{\psi}, \text{ where } \ket{\psi} \text{ is representative of } SLOCC_{\vec{\epsilon}}, A \in GL(2, \mathbb{C}), B \in GL(m, \mathbb{C}), C \in GL(n,\mathbb{C}) \} = \mathcal{O}_{\vec{\epsilon}}$. Hence, in particular, also the dimensions coincide. As the dimension of the orbit for $SLOCC_{\vec{\epsilon}}$ is strictly smaller than the dimension of the orbit for $SLOCC_{\vec{\epsilon}'}$ for any $\vec{\epsilon} \neq \vec{\epsilon}'$, we have that any such $SLOCC_{\vec{\epsilon}}$ is of measure zero and hence $SLOCC_{\vec{\epsilon}'}$ is of full measure. This completes the proof.

 \end{proof}

\section{Proof of Lemma \ref{lemma:Ldistribution}}
\label{app:Ldistribution}

Here, we restate Lemma \ref{lemma:Ldistribution} introduced in the main text in order to prove Theorem \ref{GenericToGeneric} and present a proof of the lemma.

\begingroup
\def\thetheorem{\ref{lemma:Ldistribution}}
\begin{lemma}
A state $\ket{\psi}$ in $2 \times m \times n$ with $\mathcal{P}_\psi = \bigoplus_{i=1}^{n-m} L_{\epsilon_i}$ can be transformed to a state $\ket{\phi}$ in $2 \times m \times (n-1)$ with $\mathcal{P}_\phi = \bigoplus_{i=1}^{n-m-1} L_{\epsilon'_i}$ via local operations for any $n \geq m + 2$ if the following condition holds.
There exsits $j \in \{1, \ldots, n-m-1\}$ such that for all $i \in \{1, \ldots, j-1\}$ $\epsilon_i = \epsilon_i'$ and for all $i \in \{j, \ldots, n-m-1\}$ $\epsilon_{i+1} \leq \epsilon_i'$, where we assume $(\epsilon_i)_i$ and $(\epsilon'_i)_i$ to be sorted in ascending order.
\end{lemma}
\addtocounter{theorem}{-1}
\endgroup

\begin{proof}
First of all, let us note that it suffices to prove the lemma for $j=1$, as any other case can be reduced to that case by acting trivially on the subspace of $\mathcal{P}_\psi$ which contains the blocks $L_{\epsilon_1}, \ldots, L_{\epsilon_{j-1}}$. Hence, we will assume $j=1$ in the following, which implies that $\epsilon_{i+1} \leq \epsilon_i'$ holds for all $i \in \{1, \ldots, n-m-1\}$.   
To prove the statement we will explicitly derive the local operations which perform the transformation. More precisely, we will construct an $(n-1) \times n$ matrix $C$ which has rank $n-1$ such that $\mathcal{P}_\psi C^T$ is strictly equivalent to $\mathcal{P}_\psi C^T$ in Eq. (\ref{eq:prettypencil}), where the gray areas vanish. Then we will show that the rank of this matrix pencil, $\mathcal{\tilde{P}}$, is $m$. We will then explicitly construct $\tilde{B} \in M_m$ such that $\tilde{B} \mathcal{P}_\phi = \mathcal{P}_\psi C^T$.
As the rank of $\mathcal{P}_\psi C^T$ is maximal (and coincides with the rank of the matrix pencil $\mathcal{P}_\phi$), $\tilde{B}$ must be invertible, which shows that there exist $B$ and $C$ such that $\identity \otimes B \otimes C \ket{\psi} = \ket{\phi}$.

Let us write $d = n - m$, where $d\geq 2$ and let us assume that the matrix pencils $\mathcal{P}_\psi$ and $\mathcal{P}_\phi$ are in KCF and, furthermore, w.l.o.g. the blocks $L_{\epsilon_{1}}, \ldots, L_{\epsilon_{d}}$ present in $\mathcal{P}_\psi$ and $L_{\epsilon'_{1}}, \ldots, L_{\epsilon'_{d-1}}$ present in $\mathcal{P}_\phi$ are arranged in order of ascending size. Let us now consider an operation $C$ by the third party given in terms of the $n \times (n-1)$ matrix 
\begin{align}
\label{eq:mapc}
C^T=\sum_{i=1}^{d-1} \mathcal{I}_{p_i,p'_i}(\epsilon_i+1) + \mathcal{I}_{p_{i+1},p'_{i}+\epsilon'_i-\epsilon_{i+1}}(\epsilon_{i+1}+1),
\end{align}
where $p_i = \sum_{j=1}^{i-1} (\epsilon_j + 1)$, $p'_i = \sum_{j=1}^{i-1} (\epsilon'_j + 1)$, and $\mathcal{I}_{k,l}(\epsilon)$ is an $(m+d) \times (m+d-1)$ matrix with all entries vanishing except an $\epsilon \times \epsilon$ identity submatrix with its upper left corner placed at the coordinate $(k,l)$.
This operator transforms the pencil $\mathcal{P}_\psi$ to an $m \times (n-1)$ sized pencil $\mathcal{P}_\psi C^T$ given by
\begin{equation}
\label{eq:prettypencil}
\begin{tikzpicture}[baseline=(current  bounding  box.center),
style1/.style={
  matrix of math nodes,
  every node/.append style={text width=#1,align=center,minimum height=2ex},
  nodes in empty cells,
  left delimiter=(,
  right delimiter=),
  },
]
\matrix[style1=0.15cm] (1mat)
{
  & & & & & & & & & & & & & & & & & & & & & & & & & & & & & & &  \\
  & & & & & & & & & & & & & & & & & & & & & & & & & & & & & & &  \\
  & & & & & & & & & & & & & & & & & & & & & & & & & & & & & & &  \\
  & & & & & & & & & & & & & & & & & & & & & & & & & & & & & & &  \\
  & & & & & & & & & & & & & & & & & & & & & & & & & & & & & & &  \\
  & & & & & & & & & & & & & & & & & & & & & & & & & & & & & & &  \\
  & & & & & & & & & & & & & & & & & & & & & & & & & & & & & & &  \\
  & & & & & & & & & & & & & & & & & & & & & & & & & & & & & & &  \\
  & & & & & & & & & & & & & & & & & & & & & & & & & & & & & & &  \\
  & & & & & & & & & & & & & & & & & & & & & & & & & & & & & & &  \\
  & & & & & & & & & & & & & & & & & & & & & & & & & & & & & & &  \\
  & & & & & & & & & & & & & & & & & & & & & & & & & & & & & & &  \\
  & & & & & & & & & & & & & & & & & & & & & & & & & & & & & & &  \\
  & & & & & & & & & & & & & & & & & & & & & & & & & & & & & & &  \\
  & & & & & & & & & & & & & & & & & & & & & & & & & & & & & & &  \\
  & & & & & & & & & & & & & & & & & & & & & & & & & & & & & & &  \\
  & & & & & & & & & & & & & & & & & & & & & & & & & & & & & & &  \\
  & & & & & & & & & & & & & & & & & & & & & & & & & & & & & & &  \\
  & & & & & & & & & & & & & & & & & & & & & & & & & & & & & & &  \\
  & & & & & & & & & & & & & & & & & & & & & & & & & & & & & & &  \\
  & & & & & & & & & & & & & & & & & & & & & & & & & & & & & & &  \\
  & & & & & & & & & & & & & & & & & & & & & & & & & & & & & & &  \\
  & & & & & & & & & & & & & & & & & & & & & & & & & & & & & & &  \\
  & & & & & & & & & & & & & & & & & & & & & & & & & & & & & & &  \\
  & & & & & & & & & & & & & & & & & & & & & & & & & & & & & & &  \\
};

\draw[solid]
  (1mat-1-1.north west) -- (1mat-1-3.north east);
\draw[solid]
  (1mat-2-1.south west) -- (1mat-2-3.south east);
\draw[solid]
  (1mat-1-1.north west) -- (1mat-2-1.south west);
\draw[solid]
  (1mat-1-3.north east) -- (1mat-2-3.south east);
  
\draw[solid]
  (1mat-3-2.north west) -- (1mat-3-4.north east);
\draw[solid]
  (1mat-4-2.south west) -- (1mat-4-4.south east);
\draw[solid]
  (1mat-3-2.north west) -- (1mat-4-2.south west);
\draw[solid]
  (1mat-3-4.north east) -- (1mat-4-4.south east);
  
\draw[solid]
  (1mat-3-5.north west) -- (1mat-3-7.north east);
\draw[solid]
  (1mat-4-5.south west) -- (1mat-4-7.south east);
\draw[solid]
  (1mat-3-5.north west) -- (1mat-4-5.south west);
\draw[solid]
  (1mat-3-7.north east) -- (1mat-4-7.south east);
  
\draw[solid]
  (1mat-5-6.north west) -- (1mat-5-9.north east);
\draw[solid]
  (1mat-7-6.south west) -- (1mat-7-9.south east);
\draw[solid]
  (1mat-5-6.north west) -- (1mat-7-6.south west);
\draw[solid]
  (1mat-5-9.north east) -- (1mat-7-9.south east);
  
\draw[solid]
  (1mat-5-10.north west) -- (1mat-5-13.north east);
\draw[solid]
  (1mat-7-10.south west) -- (1mat-7-13.south east);
\draw[solid]
  (1mat-5-10.north west) -- (1mat-7-10.south west);
\draw[solid]
  (1mat-5-13.north east) -- (1mat-7-13.south east);
  
\draw[solid]
  (1mat-8-10.north west) -- (1mat-8-14.north east);
\draw[solid]
  (1mat-11-10.south west) -- (1mat-11-14.south east);
\draw[solid]
  (1mat-8-10.north west) -- (1mat-11-10.south west);
\draw[solid]
  (1mat-8-14.north east) -- (1mat-11-14.south east);
  
\draw[solid]
  (1mat-16-26.north west) -- (1mat-16-31.north east);
\draw[solid]
  (1mat-20-26.south west) -- (1mat-20-31.south east);
\draw[solid]
  (1mat-16-26.north west) -- (1mat-20-26.south west);
\draw[solid]
  (1mat-16-31.north east) -- (1mat-20-31.south east);

\draw[solid]
  (1mat-21-27.north west) -- (1mat-21-32.north east);
\draw[solid]
  (1mat-25-27.south west) -- (1mat-25-32.south east);
\draw[solid]
  (1mat-21-27.north west) -- (1mat-25-27.south west);
\draw[solid]
  (1mat-21-32.north east) -- (1mat-25-32.south east);
  
\draw[dashed]
  (1mat-1-4.north east) -- (1mat-25-4.south east);
\draw[dashed]
  (1mat-1-9.north east) -- (1mat-25-9.south east);
\draw[dashed]
  (1mat-1-14.north east) -- (1mat-25-14.south east);
\draw[dashed]
  (1mat-1-25.north east) -- (1mat-25-25.south east);

\draw[dotted]
  (1mat-1-3.north east) -- (1mat-25-3.south east);
\draw[dotted]
  (1mat-1-6.north east) -- (1mat-25-6.south east);
\draw[dotted]
  (1mat-1-10.north east) -- (1mat-25-10.south east);
\draw[dotted]
  (1mat-1-15.north east) -- (1mat-25-15.south east);
\draw[dotted]
  (1mat-1-27.north east) -- (1mat-25-27.south east);

\fill [black ,opacity=0.15] (1mat-3-7.north east) rectangle (1mat-4-7.south west);
\fill [black ,opacity=0.15] (1mat-5-10.north east) rectangle (1mat-7-14.south west);
\fill [black ,opacity=0.15] (1mat-16-27.north east) rectangle (1mat-20-32.south west);    

\node[font=\large] 
  at ([xshift=6pt,yshift=-6pt]1mat-1-1.east) {$L_{\epsilon_1}$};  
\node[font=\large] 
  at ([xshift=6pt,yshift=-6pt]1mat-3-2.east) {$L_{\epsilon_2}$};  
\node[font=\large] 
  at ([xshift=6pt,yshift=-6pt]1mat-3-5.east) {$L_{\epsilon_2}$};  
\node[font=\large] 
  at ([xshift=0pt,yshift=0pt]1mat-6-7.east) {$L_{\epsilon_3}$};  
\node[font=\large] 
  at ([xshift=0pt,yshift=0pt]1mat-6-11.east) {$L_{\epsilon_3}$};  
\node[font=\large] 
  at ([xshift=6pt,yshift=-6pt]1mat-9-11.east) {$L_{\epsilon_4}$};
\node[font=\large] 
  at ([xshift=0pt,yshift=0pt]1mat-18-28.east) {$L_{\epsilon_{d-1}}$};
\node[font=\large] 
  at ([xshift=0pt,yshift=0pt]1mat-23-29.east) {$L_{\epsilon_{d}}$};

\node[font=\huge] 
  at (1mat-12-20) {$\ddots$};
  \node[font=\huge] 
  at (1mat-17-20) {$\ddots$};      

\draw[decoration={brace,raise=5pt},decorate]
  (1mat-1-1.north west) -- 
  node[above=7pt] {$\epsilon_{1} + 1$} 
  (1mat-1-3.north east);
\draw[decoration={brace,raise=5pt},decorate]
  (1mat-1-4.north west) -- 
  node[above=7pt] {$\epsilon_{2} + 1$} 
  (1mat-1-6.north east);
\draw[decoration={brace,raise=5pt},decorate]
  (1mat-1-7.north west) -- 
  node[above=7pt] {$\epsilon_{3} + 1$} 
  (1mat-1-10.north east);
\draw[decoration={brace,raise=5pt},decorate]
  (1mat-1-11.north west) -- 
  node[above=7pt] {$\epsilon_{4} + 1$} 
  (1mat-1-15.north east);
\draw[decoration={brace,raise=5pt},decorate]
  (1mat-1-28.north west) -- 
  node[above=7pt] {$\epsilon_{d}$} 
  (1mat-1-32.north east);
   
\draw[decoration={brace,mirror,raise=5pt},decorate]
  (1mat-4-2.south west) -- 
  node[below=7pt] {$\epsilon_{2} + 1$} 
  (1mat-4-4.south east);
  \draw[decoration={brace,mirror,raise=5pt},decorate]
  (1mat-7-6.south west) -- 
  node[below=7pt] {$\epsilon_{3} + 1$} 
  (1mat-7-9.south east);
  \draw[decoration={brace,mirror,raise=5pt},decorate]
  (1mat-11-10.south west) -- 
  node[below=7pt] {$\epsilon_{4} + 1$} 
  (1mat-11-14.south east);

\draw[decoration={brace,mirror,raise=5pt},decorate]
  (1mat-25-1.south west) -- 
  node[below=7pt] {$\epsilon'_{1} + 1$} 
  (1mat-25-4.south east);
\draw[decoration={brace,mirror,raise=5pt},decorate]
  (1mat-25-5.south west) -- 
  node[below=7pt] {$\epsilon'_{2} + 1$} 
  (1mat-25-9.south east);
\draw[decoration={brace,mirror,raise=5pt},decorate]
  (1mat-25-10.south west) -- 
  node[below=7pt] {$\epsilon'_{3} + 1$} 
  (1mat-25-14.south east);
\draw[decoration={brace,mirror,raise=5pt},decorate]
  (1mat-25-26.south west) -- 
  node[below=7pt] {$\epsilon'_{d-1} + 1$} 
  (1mat-25-32.south east);
   
\node at ([xshift=-30pt,yshift=-1.2pt]1mat.west) {$\mathcal{P}_{\psi} C^T =$};  
\end{tikzpicture}.
\end{equation}
In other words, the operator $C^T$ copies the blocks $L_{\epsilon_i}$ and redistributes these copies to new columns as indicated in Eq. (\ref{eq:prettypencil}). Note that this transformation is achieved by column operations on the pencil solely. Hence, indeed only the third party has to apply some operation to perform the transformation. 
Let us look at the structure of the pencil $\mathcal{P}_{\psi} C^T$ in Eq. (\ref{eq:prettypencil}) in more detail. The columns of the pencil can be grouped into $d-1$ sectors of width $\epsilon'_i + 1$, respectively. In each of this sectors, two blocks, $L_{\epsilon_i}$ and $L_{\epsilon_{i+1}}$, occur. It is ensured that these blocks entirely fit into the sector, as due to the assumption $\epsilon_{i+1} \leq \epsilon'_i$ and thus also $\epsilon_{i} \leq \epsilon'_i$. Furthermore note, that the horizontal overlap of those two blocks is at least one column, as 
\begin{align}
	\epsilon'_i = m - \sum_{j=1, j \neq i}^{d-1} \epsilon'_j = m - \sum_{j=1}^{i-1} \epsilon'_j - \sum_{j=i+1}^{d-1} \epsilon'_j \leq m - \sum_{j=1}^{i-1} \epsilon_j - \sum_{j=i+1}^{d-1} \epsilon_{j+1} = m - \sum_{j=1, j\neq i, j \neq i+1}^{d} \epsilon_j = \epsilon_i + \epsilon_{i+1},
\end{align}
where we used the assumption $\epsilon_{j+1} \leq \epsilon'_{j}$ to obtain the inequality. This shows that the transformation from $\mathcal{P}_{\psi}$ to $\mathcal{P}_{\psi} C^T$ as given in Eq. (\ref{eq:prettypencil}) is always possible. The grey shaded areas will be explained later on.

In the following, we will show that this pencil, $\mathcal{P}_{\psi} C^T$, is strictly equivalent to $\mathcal{P}_{\phi}$. To this end, we first show that $\mathcal{P}_{\psi} C^T$ is a full rank matrix pencil, and then we show that there exists some invertible operation $B$ such that $B \mathcal{P}_{\psi} C^T = \mathcal{P}_{\phi}$. 
To see that $\mathcal{P}_{\psi} C^T$ is of full rank, we consider the pencil $\mathcal{P}_{\psi} C^T$ and apply one more transformation by the third party, which is invertible, to obtain a strictly equivalent pencil $\mathcal{\tilde{P}}$ which is of lower block-triangular form. We then show that $\mathcal{\tilde{P}}$ has full rank. 
Note that $\mathcal{P}_{\psi} C^T$ has a form where the first $\epsilon_1$ rows are zero for all except the first $\epsilon_1 + 1$ columns. We now give an iterative procedure transforming $\mathcal{P}_{\psi} C^T$ to the desired pencil $\mathcal{\tilde{P}}$ which has the property that for all $i$, the first $q_{i+1}$ rows are zero on all except the first $p_{i+1}$ columns, where $q_i = \sum_{j=1}^{i-1} \epsilon_j$ and the above definition for $p_i$ is used. In particular, $\mathcal{\tilde{P}}$ is the same matrix pencil as the one given in Eq. (\ref{eq:prettypencil}), but now the entries on the grey shaded areas vanish. As observed above, $\mathcal{P}_{\psi} C^T$ has the required form for the first $q_2 = \epsilon_1$ rows. 
Let us now assume that the pencil has the required form for the first $q_{i}$ rows. We will show that the pencil can be transformed to a pencil of required form for the first $q_{i+1}$ rows. To this end, we show that the columns $p_{i} + 1$ until $p'_{i}$ can be used to cancel the entries in the rows $q_{i}+1$ until $q_{i+1}$ of columns $p_{i+1}+1$ until $p'_{i} + \epsilon_i+1$, which constitute one of the grey shaded areas in Eq. (\ref{eq:prettypencil}). To see this, note that the columns $p_{i} + 1$ until $p'_{i}$ have no non-zero entries in the first $q_i$ rows due to the assumption. On rows $q_i+1$ until $q_{i+1}$, these columns have the same entries as the columns $p_{i+1}+1$ until $p'_{i} + \epsilon_i + 1$, as these are entries of two copies of the block $L_{\epsilon_i}$ which are positioned suitably. Hence, we can subtract the columns $p_{i} + 1$ until $p'_{i}$ from the columns $p_{i+1}+1$ until $p'_{i} + \epsilon_i + 1$ in order to cancel those entries. The resulting pencil now has the required 
form for the first $q_{i+1}$ rows. Moreover, the rest of the pencil remains unchanged under this operation, because the columns $p_{i} + 1$ until $p'_{i}$ vanish on all rows after $q_{i+1}$, as the first entry on such a row appears in column $p'_{i+1} - (\epsilon_{i+1} + 1) + 1$ and we have that $(p'_{i+1} - (\epsilon_{i+1} + 1) + 1) > p'_{i}$ due to the assumption $\epsilon'_i \geq \epsilon_{i+1}$.

Let us now observe important properties of the pencil $\mathcal{\tilde{P}}$. This pencil has lower block triangular form, i.e., the pencil has rectangular blocks on its diagonal, all entries to the upper right of the blocks in the diagonal vanish, and entries in the lower left region are arbitrary.
 The first $d-1$ blocks are cyclic column permutations of the blocks $L_{\epsilon_i}$, respectively. The last block is a $\epsilon_d \times \epsilon_d$ sized block given by the last $\epsilon_d$ columns of $L_{\epsilon_d}$. Let us now show that the rank of $\mathcal{\tilde{P}}$ is maximal, i.e., $r=m$. To this end, let us delete $d-1$ columns from the pencil, one from each of the first $d-1$ rectangular blocks in the diagonal, in order to obtain a $m \times m$ matrix of lower block-triangular form whose determinant is one $m$-minor of $\mathcal{\tilde{P}}$. The determinant of a block-triangular matrix equals the product of the determinants of the blocks on the diagonal \cite{HoJo13} and as all of these determinants are non-vanishing we obtain an non-vanishing $m$-minor, which proves that $\mathcal{\tilde{P}}$ is of full rank and hence also $\mathcal{P}_{\psi} C^T$ is.

We show now that there exists an $m \times m$ matrix $\tilde{B}$ such that $\tilde{B} \mathcal{P}_\phi = \mathcal{P}_{\psi} C^T$. To see this, note that rows of the pencil $\mathcal{P}_\phi$ can be easily redistributed in order to obtain the pencil $\mathcal{P}_{\psi} C^T$. In particular, the operator $\tilde{B}$ is given by
\begin{align}
\label{eq:mapb}
\tilde{B}=\sum_{i=1}^{d-1} \mathcal{I}_{q_i,q'_i}(\epsilon_i) + \mathcal{I}_{q_{i+1},q'_{i}+\epsilon'_i-\epsilon_{i+1}}(\epsilon_{i+1}),
\end{align}
where $q'_i = \sum_{j=1}^{i-1} \epsilon'_j$. Here, $\mathcal{I}_{k,l}(\epsilon)$ is a $m\times m$ matrix which is defined analogous to before.

As we have proven that the matrix pencil $\mathcal{P}_{\psi} C^T$ is of full rank, it follows that $\tilde{B}$ is invertible as otherwise the rank of $\mathcal{P}_{\psi} C^T$ must be smaller than $m$. Defining $B = \tilde{B}^{-1}$ we have that $\mathcal{P}_{\phi} = B \mathcal{P}_{\psi} C^T$ and hence $\ket{\phi} = \identity \otimes B \otimes C \ket{\psi}$, which completes the proof.
\end{proof}

Let us prove the following Lemma about minimal indizes.
\begin{lemma}
\label{lemma:minimalindicesbound}
Given a list of $p$ linearly independent (in the sense introduced in the preliminaries) homogenous polynomial vectors in $\mu$ and $\lambda$ $(\vec{y}_1, \vec{y}_2, \ldots, \vec{y}_p)$ of ascending degrees 
and given a matrix pencil $\mathcal{P}$ such that $\mathcal{P} \vec{y}_l = 0 \ \forall \; 1 \leq  l \leq p$. Let  
$\epsilon_1 \leq \epsilon_2 \leq \ldots \leq \epsilon_p$ be the first $p$ minimal indices of $\mathcal{P}$. Then it holds that $\epsilon_l \leq \operatorname{deg}(\vec{y}_l) \ \forall \; 1 \leq  l \leq p$, where $\operatorname{deg}(\vec{y}_l)$ denotes the degree of $\vec{y}_l$.
\end{lemma}

\begin{proof}
W.l.o.g. we assume that the matrix pencil $\mathcal{P}$ is in KCF. Note that the $J$ block and the $L^T_{\nu_i}$ blocks do not have any vector in the right null-space. The vectors $\vec{y}_l$ can hence only be non-vanishing for the first $\sum_{j=1}^a (\epsilon_j + 1)$ entries, where $a$ denotes the number of right null-space blocks in $\mathcal{P}$. Moreover, for each $L_\epsilon$ there exists exactly one (linearly independent) vector in the nullspace. Note that the smallest possible degree of such a vector is $\epsilon$. Let us now write the vectors $\vec{y}_l$ as $\vec{y}_l = (\vec{y}^1_l, \ldots, \vec{y}^a_l, 0, \ldots)^T$, where the dimension of $\vec{y}^k_l$ is $\epsilon_k+1$.
Given the block diagonal form of $\mathcal{P}$ (in KCF) we have that $\mathcal{P} \vec{y}_l = \mathcal{P} (\vec{y}^1_l, \ldots, \vec{y}^a_l, 0, \ldots)^T = (L_{\epsilon_1}\vec{y}^1_l, \ldots, L_{\epsilon_a}\vec{y}^a_l, 0, \ldots)^T$.
The fact that each $L_{\epsilon_i}$ has only one vector in the null-space implies that for any set of $k$ linearly independent vectors $\{\vec{y}_l\}$ one needs to have at least $k$ different $\vec{y}_l^k \neq 0$. 
As mentioned before, for any $\vec{y}_l^k$ for which $L_{\epsilon_k} \vec{y}_l^k = 0$ it must hold that $\operatorname{deg}(\vec{y}_l^k) \geq \epsilon_k$ and therefore $\operatorname{deg}(\vec{y}_l) \geq \epsilon_k$. As the vectors $\left(\vec{y}_l\right)$ are sorted in order of increasing degree, the assertion $\operatorname{deg}(\vec{y}_l) \geq \epsilon_l$ for all $l$ follows.
\end{proof}

\section{Example of an optimal common resource state of $2\times 4\times 4$ states in $2\times 4\times 6$}
\label{sec:crex}

In order to illustrate the results and methods derived in Section \ref{sec:cr} we present here an example of an optimal common resource state, $\ket{\Omega}$, in $2 \times 4 \times 6$. We show here that 
\begin{equation}
\left|\Omega\right\rangle =\left|100\right\rangle +\left|001\right\rangle +\left|112\right\rangle +\left|013\right\rangle +\left|123\right\rangle +\left|024\right\rangle +\left|135\right\rangle 
\end{equation}
can be probabilistically transformed to any $2 \times 4 \times 4$ state.
Note that this state corresponds to the matrix pencil given in Eq. (\ref{eq:mmocr}). We have proven that $\ket{\Omega}$ is a common resource state for $2 \times 4 \times 4$ states in Theorem \ref{theo:mmcr} and shown its optimality in Theorem \ref{theo:mmocr}. 
Here, we explicitly show another possible way of obtaining any $2 \times 4 \times 4$ state through the intermediate step of $2 \times 4 \times 5$ states. Considering all possible KCF of matrix pencils corresponding to states in $2 \times 4 \times 4$ we obtain the following KCF representing all SLOCC classes in $2\times 4 \times 4$ (see also \cite{ChCh06})
\begin{equation}
\begin{array}{cccccccc}
 & \mathcal{P}_{1}(x)=M^1(0)\oplus M^1(1)\oplus M^1(x)\oplus N^1 &  &  &  & \mathcal{P}_{9}=M^2(0)\oplus M^1(1)\oplus N^1\\
 & \mathcal{P}_{2}=M^1(0)\oplus M^1(0) \oplus M^1(1)\oplus N^1 &  &  &  & \mathcal{P}_{10}=M^2(0)\oplus M^1(0)\oplus M^1(0)\\
 & \mathcal{P}_{3}=M^1(0)\oplus M^1(0) \oplus M^1(0)\oplus M^1(1) &  &  &  & \mathcal{P}_{11}=M^2(0)\oplus M^2(0)\\
 & \mathcal{P}_{4}=M^2(0) \oplus M^1(0) \oplus M^1(1) &  &  &  & \mathcal{P}_{12}=M^3(0)\oplus M^1(0)\\
 & \mathcal{P}_{5}=M^3(0)\oplus M^1(1) &  &  &  & \mathcal{P}_{13}=M^4(0)\\
 & \mathcal{P}_{6}=M^1(0)\oplus M^1(0) \oplus M^1(1)\oplus M^1(1) &  &  &  & \mathcal{P}_{14}=L_{1}\oplus L_{1}^{T}\oplus M^1(0)\\
 & \mathcal{P}_{7}=M^2(0)\oplus M^1(1)\oplus M^1(1) &  &  &  & \mathcal{P}_{15}=L_{1}\oplus L_{2}^{T}\\
 & \mathcal{P}_{8}=M^2(0)\oplus M^2(1) &  &  &  & \mathcal{P}_{16}=L_{2}\oplus L_{1}^{T}.
\end{array}
\end{equation}

In Fig. \ref{fig:mappings}, we explicitly present matrix pencils corresponding to a possible set of intermediate states in $2 \times 4 \times 5$ (which are all probabilistically reachable from $\ket{\Omega}$). Then, any $2 \times 4 \times 4$ state can be reached from at least one of the intermediate states as indicated by the transformations marked by arrows in Figure \ref{fig:mappings}. It can be easily verified that the marked transformations are indeed possible. Where applicable, we indicate the lemmata which can be used to prove this fact next to the arrows. 
\begin{figure}[!htb]
\includegraphics[scale=0.8]{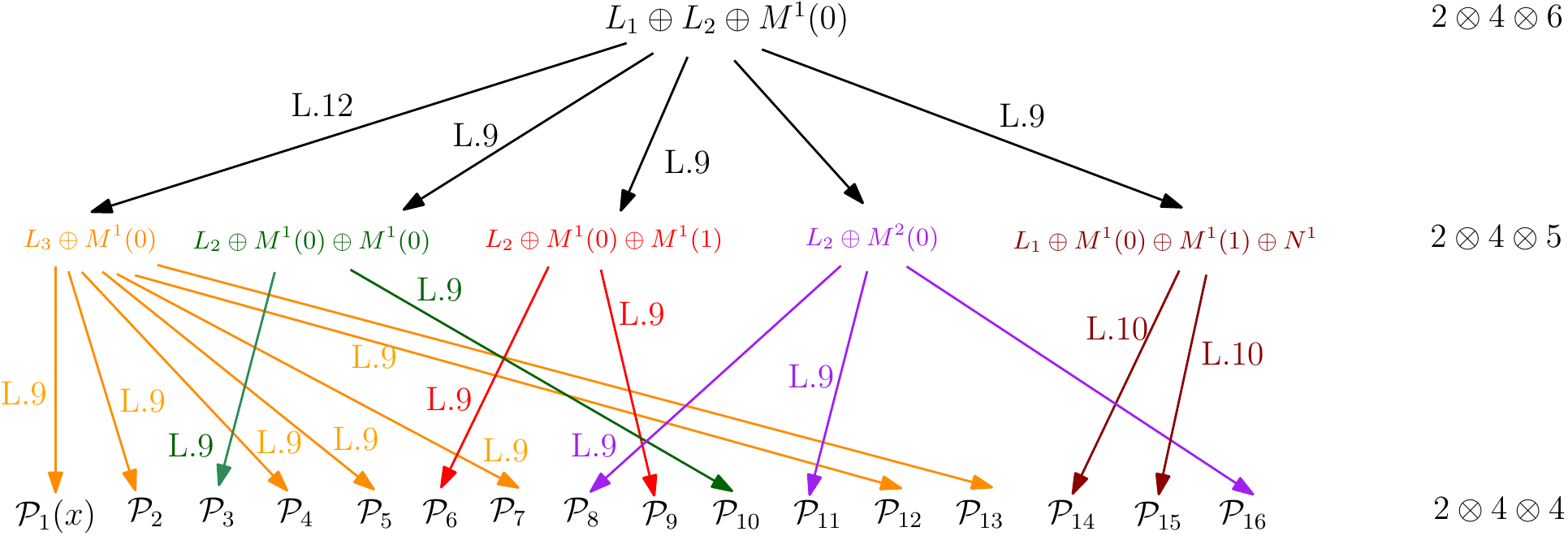}
\caption{Optimal common resource state of $2\times 4\times 4$ states in $2\times 4\times 6$. We present one possible common resource state in $2\times 4\times 6$ and show how any $2\times 4\times 4$ state can be reached through intermediate states in $2\times 4\times 5$. Where applicable, the labels next to the arrows indicate the lemmata which can be used to prove that the marked transformations are indeed possible. Note however that all transformations can be easily computed (without actually using the lemmata). Apart from the depicted transformations there exist many more. However, the ones shown here suffice to prove that the state $\ket{\Omega}$ is a common resource for all states in $\mathbb{C}^2 \otimes \mathbb{C}^4 \otimes \mathbb{C}^4$}
\label{fig:mappings}
\end{figure}

\end{document}